\documentclass[envcountsect, 11pt, fleqn]{llncs}

\usepackage[left=3cm,right=3cm,top=3.5cm,bottom=3.5cm]{geometry}
\usepackage{cite}
\usepackage{subfigure}
\usepackage{amsmath}
\usepackage{amsfonts}
\usepackage{tikz}
\usepackage{mathdots}
\usetikzlibrary{shapes.misc}
\usepackage[ruled]{algorithm2e}
\usepackage[citecolor=blue, linkcolor=blue, colorlinks=true]{hyperref}
\usepackage{pgfplots}
\usepgfplotslibrary{groupplots}

\renewcommand{\tabcolsep}{0.5em}

\newcommand\refTheorem[1]{Theorem~\hyperref[#1]{\ref*{#1}}}
\newcommand\refLemma[1]{Lemma~\hyperref[#1]{\ref*{#1}}}
\newcommand\refCorollary[1]{Corollary~\hyperref[#1]{\ref*{#1}}}
\newcommand\refFigure[1]{Figure~\hyperref[#1]{\ref*{#1}}}
\newcommand\refAlgorithm[1]{Algorithm~\hyperref[#1]{\ref*{#1}}}
\newcommand\refEquation[1]{(\hyperref[#1]{\ref*{#1}})}
\newcommand\refExample[1]{Example~\hyperref[#1]{\ref*{#1}}}
\newcommand\refSection[1]{Section~\hyperref[#1]{\ref*{#1}}}
\newcommand\refDefinition[1]{Definition~\hyperref[#1]{\ref*{#1}}}

\newcommand\expected{\mathbb{E}}

\newcommand\supp{\text{supp}}

\newcommand\strats{\ensuremath{\{1,2\}}}
\newcommand\TV{{\rm TV}}

\newcommand\complexityclass[1]{\textbf{#1}}

\begin{document}
\title{Query Complexity of Approximate Equilibria in Anonymous Games}




\author{Paul W. Goldberg\inst{1} \and Stefano Turchetta\inst{2}}



\institute{University of Oxford \email{paul.goldberg@cs.ox.ac.uk} \and University of Oxford \email{stefano.turchetta@cs.ox.ac.uk}}


\maketitle
\setcounter{footnote}{0}

\begin{abstract}
We study the computation of equilibria of anonymous games,
via algorithms that may proceed via a sequence of adaptive queries to
the game's payoff function, assumed to be unknown initially.
The general topic we consider is \emph{query complexity}, that is, how many
queries are necessary or sufficient to compute an exact or approximate Nash
equilibrium.

We show that exact equilibria cannot be found via query-efficient
algorithms. We also give an example of a 2-strategy, 3-player anonymous
game that does not have any exact Nash equilibrium in rational numbers.
However, more positive query-complexity bounds are attainable if either further
symmetries of the utility functions are assumed or we focus on approximate equilibria.
We investigate four sub-classes of anonymous games previously considered by
\cite{bfh09, dp14}.

Our main result is a new randomized query-efficient
algorithm that finds a $O(n^{-1/4})$-approximate
Nash equilibrium querying $\tilde{O}(n^{3/2})$
payoffs and runs in time $\tilde{O}(n^{3/2})$. This improves on the running time
of pre-existing algorithms for approximate equilibria of anonymous games,
and is the first one to obtain an inverse polynomial approximation in poly-time.
We also show how this can be utilized as an efficient polynomial-time
approximation scheme (PTAS).
Furthermore, we prove that $\Omega(n \log{n})$ payoffs must
be queried in order to find any $\epsilon$-well-supported Nash
equilibrium, even by randomized algorithms.
\end{abstract}


\section{Preliminaries}

This paper studies {\em anonymous} games,
in which a large number of players $n$ share a constant number of
pure strategies, and the payoff to a player depends on the number
of players who use each strategy, but not their identities.
Due to this property, these games have a polynomial (in $n$) size
representation. \cite{dp14} consider anonymous games and graphical
games to be the two most important classes of concisely-represented
multi-player games. Anonymous games appear frequently in practice, for example
in voting systems, traffic routing, or auction settings.
Although they have polynomial-sized representations, the
representation may still be inconveniently large, making it desirable
to work with algorithms that do not require all the data on a particular game
of interest.

Query complexity is motivated in part by the observation
that a game's entire payoff function may be syntactically cumbersome.
It also leads to new results that distinguish the difficulty of alternative
solution concepts. We assume that an algorithm has black-box access to the
payoff function, via queries that specify an anonymized profile and return 
one or more of the players' payoffs.

\subsection{Anonymous Games}

A $k$-strategy anonymous game is a tuple $(n, k, \{u^i_j\}_{i \in [n],
j \in [k]})$ that consists of $n$ players, $k$ pure
strategies per player, and a utility function $u^i_j : \Pi_{n - 1}^k
\longrightarrow [0, 1]$ for each player $i \in [n]$ and every strategy
$j \in [k]$, whose domain is the set $\Pi_{n - 1}^k := \{(x_1, \dots,
x_k) \in \bbbn_0^k : \sum_{j \in [k]} x_j = n-1\}$ of all possible ways
to partition $n - 1$ players into the $k$ strategies. The number of
payoffs stored by any such game is $n \cdot \left|\Pi_{n - 1}^k\right| = O\left(n^k\right)$, i.e.,
polynomial in $n$ if $k$ is a constant, which we will always assume throughout this paper.
In the special case of $k=2$, we regard $u^i_j$'s input as being the number of players other than $i$ that play strategy 1. The number of payoffs stored by a $2$-strategy game clearly
is then $2n^2$. As indicated by $u^i_j$'s codomain, we make a standard
assumption that all payoffs are normalized into the interval $[0, 1]$.

Let $e_j$ denote the unit vector of length $k$ with 1 at its $j$-th component.
A mixed strategy of player $i$ is represented using a random vector
$\mathcal{X}_i$, which takes value $e_j$ with the probability that $i$ plays $j$.
Letting $p_j^i$ be the probability that player $i$ plays strategy $j$, we then have
$\expected[\mathcal{X}_i] = \left(p^i_1, \dots, p^i_k\right)$.
Let $\mathcal{X}_{-i} := \sum_{\ell \in [n] \setminus \{i\}} \mathcal{X}_\ell$ be the sum of
$n-1$ such random vectors, where the subscript $-i$ denotes all players other than $i$.
The expected utility obtained
by player $i \in [n]$ for playing strategy $j \in [k]$ against $\mathcal{X}_{-i}$ is
\[
\expected[u^i_j(\mathcal{X}_{-i})] :=
\sum_{x \in \Pi_{n - 1}^k} u^i_j(x) \cdot \text{Pr}[\mathcal{X}_{-i} = x].
\]
Let $\mathcal{X} := (\mathcal{X}_i, \mathcal{X}_{-i})$.
If $i$ is playing a mixed strategy $(p^i_1, \dots, p^i_k)$,
her expected payoff simply consists of a weighted average, i.e.,
\[
\expected[u^i(\mathcal{X})] := \sum_{j = 1}^k p_j^i \cdot
\expected[u^i_j(\mathcal{X}_{-i})].
\]

For two-strategy games, the strategy played by player $i$ can be represented
by a Bernoulli random variable $X_i$ indicating whether or not $i$ plays strategy 1.
Hence, a mixed strategy for $i$ simply is the probability $p_i := \expected[X_i]$
that $i$ plays strategy 1. Similarly to the $k$-strategy case, let $X_{-i} := \sum_{\ell \in [n]
\setminus \{i\}} X_\ell$ be the sum of all the random
variables other than $X_i$. The expected utility obtained
by player $i \in [n]$ for playing pure strategy $j \in \strats$ against $X_{-i}$ is
\[
\expected[u^i_j(X_{-i})] := \sum_{x = 0}^{n-1} u^i_j(x) \cdot \Pr[X_{-i} = x].
\]
The probability mass function (p.m.f.) of $\mathcal{X}_{-i}$ and $X_{-i}$ respectively
are a Poisson Multinomial Distribution and a Poisson Binomial
Distribution. Both of them can be computed in polynomial time using dynamic programming (see e.g., \cite{dp14}),
i.e., expected utilities are computable in polynomial time.

\subsection{Exact and Approximate Nash Equilibria.}

With the above notation, we say that $\mathcal{X}_i$ is a
best-response if and only if $\expected[u^i(\mathcal{X})] \geq
\expected[u^i_j(\mathcal{X}_{-i})]$ for all $j \in [k]$. A \emph{Nash
  equilibrium} (NE) requires the players to be best-responding to each
other; therefore, the above best-response condition must hold for
every $i \in [n]$. This can be also viewed as no player having an
incentive to deviate from her strategy. We consider a
relaxation of NE, the notion of an
\emph{$\epsilon$-approximate Nash equilibrium} ($\epsilon$-NE), where
every player's incentive to deviate is at most $\epsilon > 0$. We say
that $(\mathcal{X}_i)_{i \in [n]}$, which represents a mixed-strategy
profile, constitutes an $\epsilon$-NE if for all $i \in [n]$ and all
$j \in [k]$,
\[
\expected[u^i(\mathcal{X})] +
\epsilon \geq \expected[u^i_j(\mathcal{X}_{-i})].
\]
This definition, however, does not prohibit allocating a small amount
of probability to arbitrarily bad strategies. An
\emph{$\epsilon$-approximate well-supported Nash equilibrium}
($\epsilon$-WSNE) addresses this issue by forcing every player to
place a positive amount of probability solely on
$\epsilon$-approximate best-responses, i.e., $(\mathcal{X}_i)_{i \in [n]}$
constitutes an $\epsilon$-WSNE if for all $i \in [n]$, all $j \in [k]$,
and all $\ell \in \supp(\expected[\mathcal{X}_i])$\footnote{Given a vector $v$, $\supp(v)$ is used to denote the support of $v$, i.e., $\supp(v) = \{i : v_i > 0\}$.},
\[
\expected[u^i_\ell(\mathcal{X}_{-i})] + \epsilon \geq \expected[u^i_j(\mathcal{X}_{-i})].
\]
Although an $\epsilon$-WSNE is also an $\epsilon$-NE, the converse need not be true.

\subsubsection{Sub-classes of Anonymous Games}

Motivated by our first negative result, which states that finding an exact NE in an
anonymous game may require querying a number of payoffs that is proportional
to the size of the game, we consider further restrictions of the players' utility function
(from \cite{bfh09, dp14}) with the aim of getting more positive results for these.

An anonymous game is \emph{symmetric} if for all $i, \ell \in [n]$, all $j
\in [k]$, and all $x \in \Pi_{n - 1}^k$, then $u^i_j(x) = u^\ell_j(x)$,
i.e., all players share the same utility function. 

An anonymous game is \emph{self-anonymous} if for all $i \in [n]$, all $j, \ell \in
[k]$, and all $x \in \{y \in \Pi_{n - 1}^k : y_\ell \neq 0\}$,
then $u^i_j(x) = u^i_\ell(x + e_j - e_\ell)$,
i.e., player $i$'s preferences depend on how all the $n$
players are partitioned into the $k$ strategies; therefore, $i$ does not
distinguish herself from the others. 

An anonymous game is \emph{self-symmetric} if it is both symmetric and
self-anonymous.

An anonymous game is \emph{Lipschitz} if every player's utility function
is Lipschitz continuous, in the sense that for all $i \in [n]$, $j \in [k]$, and all $x, y \in \Pi_{n - 1}^k$,
$\left|u^i_j(x) - u^i_j(y)\right| \leq \lambda \left\| x - y \right\|_1$, where $\lambda \geq 0$
is the so-called Lipschitz constant and $\| \cdot \|_1$ denotes the $L_1$ norm.

\subsection{Query-efficiency and Payoff Query Models.}

Our general interest is in polynomial-time algorithms that find
solutions of anonymous games, while checking just a small fraction of
the $O(n^k)$ payoffs of an $n$-player, $k$-strategy game. The basic
kind of query is a {\em single-payoff query} which receives as input a
player $i \in [n]$, a strategy $j \in [k]$, and a partition of $n-1$ players into
the $k$ strategies $x \in \Pi_{n-1}^k$, and it
returns the corresponding payoff $u^i_j(x)$.

The \emph{query complexity} of an algorithm is the expected number of single-payoff
queries that it needs in the worst case. Hence, an algorithm is
query-efficient if its query complexity is $o(n^k)$.

A {\em profile query} (used in \cite{fggs13}) consists of an action profile
$(a_1, \dots, a_n) \in [k]^n$ as input and outputs the payoffs that
{\em every} player $i$ obtains according to that profile.
Clearly, any profile query can be simulated by a sequence of $n$ single-payoff queries.
Finally, an {\em all-players query} consists of a pair $(j, x)$ for $j\in [k]$, $x \in \Pi_{n-1}^k$,
and the response to $(j, x)$ is the vector of all players' utilities $(u^1_j(x), \ldots, u^n_j(x))$.

We consider the cost of a query to be
equal to the number of payoffs it returns; hence, a profile or an all-players
query costs $n$ single-payoff queries. We find that an algorithm being
constrained to utilize profile queries may incur a linear loss in
query-efficiency (cf. \refSection{section:comparison_of_query_models}).
Therefore, we focus on single-payoff and all-players queries,
which better exploit the symmetries of anonymous games.

\subsection{Related Work}

In the last decade, there has been interest in the complexity
of computing approximate Nash equilibria.
A main reason is the \complexityclass{PPAD}-completeness results for
computing an exact NE, for normal-form games \cite{dgp09,cdt09}
(the latter paper extends the hardness also to fully polynomial-time
approximation schemes (FPTAS)),
and recently also for anonymous games with $7$ strategies \cite{cdo14}.
The \complexityclass{FIXP}-completeness results of
\cite{EY10} for multiplayer games show an algebraic obstacle to the
task of writing down a useful description of an exact equilibrium.
On the other hand, there exists a subexponential-time algorithm to
find an $\epsilon$-NE in normal-form games \cite{lmm03},
raising the well-known open question
of the possible existence of a PTAS for these games.

Daskalakis and Papadimitriou proved that anonymous games admit a PTAS and
provided several improvements of its running time over the past few years.
Their first algorithm \cite{dp07}
concerns two-strategy games and is based upon the quantization of
the strategy space into nearby multiples of $\epsilon$.
This result was also extended to the multi-strategy case
\cite{dp08discretized}.
\cite{daskalakis08} subsequently gave an efficient PTAS whose
running time is $\text{poly}(n) \cdot (1 / \epsilon)^{O(1 /
  \epsilon^2)}$, which relies on a better understanding of
the structure of $\epsilon$-equilibria in two-strategy anonymous
games: There exists an $\epsilon$-WSNE where either a
small number of the players -- at most $O(1 / \epsilon^3)$ --
randomize and the others play pure strategies, or whoever randomizes
plays the same mixed strategy.
Furthermore, \cite{dp09}
proved a lower bound on the running time needed by any {\em oblivious}
algorithm, which lets the latter algorithm be essentially optimal. In
the same article, they show that the lower bound can be broken by
utilizing a non-oblivious algorithm, which has the currently best-known
running time for finding an $\epsilon$-equilibrium in
two-strategy anonymous games of $O(\text{poly}(n) \cdot (1/
\epsilon)^{O(\log^2(1 / \epsilon))})$. A complete proof is in \cite{dp13sparse}.
It has been recently shown that also $k$-strategy anonymous games admit
an efficient PTAS \cite{dkt15, dks15}.

In \refSection{section:lipschitz_games} we present a bound for
$\lambda$-Lipschitz games, in which $\lambda$ is a parameter limiting the
rate at which $u^i_j(x)$ changes as $x$ changes.
Any $\lambda$-Lipschitz $k$-strategy anonymous game is guaranteed to have an
$\epsilon$-approximate {\em pure} Nash equilibrium, with $\epsilon =
O(\lambda k)$ \cite{as13,dp14}.
The convergence rate to a Nash equilibrium of best-reply dynamics in the context
of two-strategy Lipschitz anonymous games is studied by \cite{kfh09,bab13}.
Moreover, \cite{bfh09}
showed that finding a pure equilibrium in anonymous games is easy if the number
of strategies is constant w.r.t. the number of players $n$, and hard
as soon as there is a linear dependence.

In the last two years, several researchers obtained
bounds for the query complexity for approximate equilibria in different game
settings, which we briefly survey. \cite{fggs13} presented
the first series of results: they studied bimatrix games,
graphical games, and congestion games on graphs.
Similar to our negative result for exact equilibria of anonymous games, it was shown
that a Nash equilibrium in a bimatrix game with $k$ strategies per player requires
$k^2$ queries, even in zero-sum games. However, more positive results arise if we move to
$\epsilon$-approximate Nash equilibria. Approximate equilibria of bimatrix games
were studied in more detail in \cite{fs13}.

The query complexity of equilibria of $n$-player games
-- a setting where payoff functions are exponentially-large -- was
analyzed in \cite{hn13,babichenko14,gr13,cct15}. \cite{hn13}
showed that exponentially many deterministic queries are required to
find a $\frac{1}{2}$-approximate correlated equilibrium (CE) and that any randomized
algorithm that finds an exact CE needs $2^{\Omega(n)}$ expected cost.
Notice that lower bounds on correlated equilibria automatically
apply to Nash equilibria. \cite{gr13} investigated in more detail the
randomized query complexity of $\epsilon$-CE and of the more demanding
$\epsilon$-well-supported CE. \cite{babichenko14}
proved an exponential-in-$n$ randomized lower
bound for finding an $\epsilon$-WSNE in $n$-player, $k$-strategy games,
for constant $k = 10^4$ and $\epsilon = 10^{-8}$.
Finally, \cite{cct15} showed (for small enough positive $\epsilon$)
an exponential lower bound for $\epsilon$-Nash equilibria of binary-action games.
These exponential lower bounds do not hold in anonymous games,
which can be fully revealed with a polynomial number of queries.

\subsection{Our Results and their Significance}

Query-efficiency serves as a criterion for distinguishing exact
from approximate equilibrium computation.
It applies to games having exponentially-large representations~\cite{hn13},
also for games having poly-sized representations
(e.g. bimatrix games~\cite{fggs13,fs13}).
Here we extend this finding to the important class of anonymous games.

We prove that even in two-strategy anonymous (and also self-anonymous) games,
an exact Nash equilibrium demands querying the payoff function exhaustively,
even with the most powerful query model (\refTheorem{thm:exactNE_lowerbound}).
Alongside this, we provide an example of a three-player, two-strategy
anonymous game whose unique Nash equilibrium needs all players to
randomize with an irrational amount of probability (\refTheorem{thm:irrational_nash}),
answering a question posed in \cite{dp14}.
These results motivate our subsequent focus on sub-classes of anonymous game and
approximate equilibria.

\refSection{section:symmetric_games} focuses on the search for pure exact equilibria,
in subclasses of anonymous games where these are guaranteed to exist.
We give tight bounds for two-strategy symmetric games, also $k$-strategy
self-symmetric games, in the latter case via connecting them with
the query complexity of searching for a local optimum of a function on
a grid graph, allowing us to apply pre-existing results for that problem.

\refSection{section:self-anonymous_games} is concerned with mixed-strategy equilibria of self-anonymous games. Although searching for exact equilibria
requires the utility function to be queried exhaustively, surprisingly, a uniform randomization over all the strategies is always a $O(n^{-1/2})$-WSNE, for games having a constant number of strategies. Therefore, no payoff needs to be queried for this. Moreover, we give a reduction that maps any general two-strategy anonymous game $G$ to a two-strategy self-anonymous game $G'$ such that if an FPTAS for $G'$ exists, then there also exists one for $G$. The possible existence of a FPTAS is the main open algorithmic question in the context of anonymous games, thus we show that the search for an answer to this question can focus on self-anonymous games (in the two-strategy case).

In addition, we exhibit a simple query-efficient algorithm that finds an approximate
pure Nash equilibrium in two-strategy Lipschitz games
(\refAlgorithm{algorithm:approxNELip}; \refTheorem{thm:lipschitz2strategy}),
which is subsequently used by our main algorithm for anonymous games.

Our main result (\refTheorem{thm:approximateNE}) is a new randomized
approximation scheme\footnote{To make \refTheorem{thm:approximateNE} easier
to read, we state it only for the best attainable approximation (i.e.,
$n^{-1/4}$); however, it is possible to set parameters
to get any approximation $\epsilon \geq n^{-1/4}$.
For details, see the proof of \refTheorem{thm:approximateNE}.}
for two-strategy anonymous games that differs
conceptually from previous ones and offers new performance guarantees.
It is query-efficient (using $o(n^2)$ queries)
and has improved computational efficiency.
It is the first PTAS for anonymous games that is polynomial in a setting
where $n$ and $1/\epsilon$ are polynomially related.
In particular, its runtime is polynomial in $n$ in a setting where
$1/\epsilon$ may grow in proportion to $n^{1/4}$ and also has an improved
polynomial dependence on $n$ for all $\epsilon \geq n^{-1/4}$.
In more detail, for any $\epsilon \geq n^{-1/4}$,
the algorithm adaptively finds a $O\left(\epsilon \right)$-NE
with $\tilde{O}\left(\sqrt{n}\right)$
(where we use $\tilde{O}(\cdot)$ to hide polylogarithmic factors)
all-players queries (i.e., $\tilde{O}\left(n^{3/2}\right)$ single payoffs)
and runs in time $\tilde{O}\left(n^{3/2}\right)$.
The best-known algorithm of~\cite{dp14} runs in time $O(\text{poly}(n) \cdot (1/
\epsilon)^{O(\log^2(1 / \epsilon))})$, where $\text{poly}(n) \geq O(n^7)$.

In addition to this, we derive a randomized logarithmic lower bound on the
number of all-players queries needed to find any non-trivial $\epsilon$-WSNE in
two-strategy anonymous games (\refTheorem{thm:log_lowerbound}).

\section{Comparison of Query Models}
\label{section:comparison_of_query_models}

We begin by comparing the profile query model with the
other two in the context of anonymous games. Recall that one
profile or one all-players query costs $n$ single-payoff queries due to returning
$n$ payoffs. We use PR, AP, and SP to denote profile, all-players,
and single-payoff queries, respectively.

\subsection{Simulating Profile Queries}
\label{section:SimulatingProfileQueries}

A PR query can be simulated by a sequence of
$n$ SP queries by requesting for each player the payoff associated
with the strategy profile of the PR query.

A PR query can be simulated using at most $k$ AP queries, as follows.
Let $a = (a_1, \dots, a_n) \in [k]^n$ be a given profile query.
Let $N_j \subseteq [n]$ be the players who play $j$ in $a$.
For all $j \in [k]$ such that $N_j \neq \emptyset$, we query
the payoff for strategy $j$ against the partition $(|N_1|, \dots,
|N_j| - 1, \dots, |N_k|) \in \Pi_{n - 1}^k$. The cost, therefore, increases
at most by a factor of $k$.

\subsection{Simulating Single-payoff Queries}

A sequence of SP queries can obviously be simulated by a sequence of
PR or AP queries having the same length. However, this approach increases
the total cost by a factor of $n$ since we receive the payoffs of every player.
In general, it cannot be simulated by a shorter sequence, for example if every
SP query is based on a different partition in $\Pi_{n - 1}^k$.

\subsection{Simulating All-players Queries}

Clearly, any AP query can be simulated by $n$ SP queries.
We consider whether AP queries can be efficiently
simulated by PR queries, and argue that this is generally not the case.
Clearly, $n$ profile queries suffice to obtain all the information returned by one
AP query, and there exist examples where this upper bound is required.
In section~\ref{sec:qclb}, we present such an example.

However, there are also cases in which an all-player query can be
simulated by a constant number of profile queries. For instance,
suppose we are dealing with a two-strategy anonymous game, and let
$\alpha \in (0, 1)$ be a constant. If an AP query retrieves the payoff for playing
strategy $j \in \strats$ when $\alpha n$ players (in total) are playing strategy
$j$, then $1 / \alpha$ PR queries are enough to get
all the information. Consider the following sequence of PR queries
when $j = 1$. The case $j=2$ is similar.
\[
(\underbrace{1, \dots, 1}_{\alpha n},
\underbrace{2, \dots, 2}_{(1 - \alpha) n}),
(\underbrace{2, \dots, 2}_{\alpha n},
\underbrace{1, \dots, 1}_{\alpha n},
\underbrace{2, \dots, 2}_{(1 - 2\alpha) n}), \dots,
(\underbrace{2, \dots, 2}_{(1 - \alpha) n},
\underbrace{1, \dots, 1}_{\alpha n}).
\]
Clearly, every query lets $\alpha n$ players know their payoffs;
therefore, $1 / \alpha$ PR queries suffice to simulate an
all-players query as specified above.
The above simulation suggests that if an AP query algorithm asks the
payoffs for strategy $j$ subject to the constraint that the fraction of
players using $j$ is bounded away from 0,
then it can be simulated by a PR query-algorithm
whose cost is only increased by a constant factor.

\subsection{A Quadratic Cost Lower Bound in the Profile Model}\label{sec:qclb}

We show that finding an $\epsilon$-WSNE, for any $\epsilon < 1 / 2$,
in a two-strategy anonymous game may require a linear number of
profile queries, i.e., a cost of $n^2$. Due to the fact that a
two-strategy anonymous game incorporates $2n^2$ payoffs, such a lower
bound rules out the possibility of a profile-query efficient algorithm
for well-supported approximate equilibria.

\begin{example}\label{ex:profile_bad}
Let $\mathcal{D}_n$ be the following distribution over two-strategy
$n$-player anonymous games where every player's payoff takes a value
in $\{0, \frac{1}{2}, 1\}$. Let a player $h \in [n]$ be chosen
uniformly at random. Let $h$'s and the other players' ($\ell$
denotes a typical player different from $h$) payoffs be defined as in
\refFigure{fig:profile_lb_game}.
\end{example}

\begin{figure}[h!]
\centering
\subfigure[$h$'s payoff]{
\begin{tabular}{ l | c | c | c | c |}
\multicolumn{1}{c}{$x$}
 &  \multicolumn{1}{c}{$0$}
 & \multicolumn{1}{c}{$\dots$}
 & \multicolumn{1}{c}{$n - 2$}
 & \multicolumn{1}{c}{$n - 1$}\\
\cline{2-5}
$u^h_1(x)$ & $\frac{1}{2}$ & $\dots$ & $\frac{1}{2}$ & $\frac{1}{2}$\\
\cline{2-5}
$u^h_2(x)$ & $0$ & $\dots$ & $0$ & $1$\\
\cline{2-5}
\multicolumn{3}{c}{}
\end{tabular}
}
\hspace{2em}
\subfigure[Other players' payoff]{
\begin{tabular}{ l | c | c | c | c |}
\multicolumn{1}{c}{$x$}
 &  \multicolumn{1}{c}{$0$}
 & \multicolumn{1}{c}{$\dots$}
 & \multicolumn{1}{c}{$n - 2$}
 & \multicolumn{1}{c}{$n - 1$}\\
\cline{2-5}
$u^{\ell}_1(x)$ & $\frac{1}{2}$ & $\dots$ & $\frac{1}{2}$ & $\frac{1}{2}$\\
\cline{2-5}
$u^{\ell}_2(x)$ & $0$ & $\dots$ & $0$ & $0$\\
\cline{2-5}
\multicolumn{3}{c}{}
\end{tabular}
}
\caption{Definition of $\mathcal{D}_n$'s payoffs.
$x$ denotes the number of players who play strategy 1.}
\label{fig:profile_lb_game}
\end{figure}

\begin{remark}
Any randomized profile-query algorithm
needs to make $\Omega(n)$ queries, which cost $\Omega(n^2)$ payoffs,
to find an $\epsilon$-WSNE of $\mathcal{D}_n$, for any
$\epsilon \in [0, \frac{1}{2})$.
\end{remark}

\begin{proof}[Proof sketch]
We use Yao's minimax principle. Clearly, every $\ell$ is always
better off playing $1$. As a result, the hidden player $h$ should play $2$.
This is the unique exact NE and $\epsilon$-WSNE, for $\epsilon < 1/2$.
Evidently, a profile-query algorithm must identify $h$ --- the unique player who plays
$2$ --- by querying vectors of the form $(1, \ldots, 1, 2, 1, \ldots, 1) \in \strats^n$.
If the $2$ does not appear at $h$'s index, the query returns no useful information.
Due to $h$ being chosen uniformly at random, a linear number of profile queries is
required.
\end{proof}

\begin{remark}
The all-players query $(2, n-1) \in \strats \times \{0, \ldots, n-1\}$ suffices to discover a pure
Nash equilibrium of a game coming from $\mathcal{D}_n$.
\end{remark}

\section{Exact Nash Equilibria}
\label{section:exact_nash_equilibria}

We lower-bound the number of single-payoff queries (the
least constrained query model) needed to find an exact NE in an
anonymous game. We exhibit games in
which any algorithm must query most of the payoffs in order
to determine what strategies form a NE. Difficult
games are ones that only possess NE in which $\Omega(n)$ players
must randomize.

\begin{example}
\label{ex:majmin}
Let $G$ be the following two-strategy, $n$-player anonymous game.
Let $n$ be even. Half of the players have a
utility function as shown by the top side~(a) of
\refFigure{fig:majority_minority_payoff}, and the remaining half as at~(b).

\begin{figure}[h!]
\centering
\resizebox{\textwidth}{!}{
\subfigure[Payoff table for ``majority-seeking'' player $i$]{
\begin{tabular}{ r | c | c | c | c | c | c |}
\multicolumn{1}{c}{$x$}
 & \multicolumn{1}{c}{$0$}
 & \multicolumn{1}{c}{$1$}
 & \multicolumn{1}{c}{$2$}
 & \multicolumn{1}{c}{$\dots$}
 & \multicolumn{1}{c}{$n-2$}
 & \multicolumn{1}{c}{$n - 1$}\\
\cline{2-7}
$u^i_1(x)$
 & $\frac{1}{2}-\left(\frac{1}{2}-\frac{1}{2n}\right)$
 & $\frac{1}{2}-\left(\frac{1}{2}-\frac{3}{2n}\right)$
 & $\frac{1}{2}-\left(\frac{1}{2}-\frac{5}{2n}\right)$
 & $\dots$
 & $\frac{1}{2}+\left(\frac{1}{2}-\frac{3}{2n}\right)$
 & $\frac{1}{2}+\left(\frac{1}{2}-\frac{1}{2n}\right)$\\
\cline{2-7}
$u^i_2(x)$
 & $\frac{1}{2}$
 & $\frac{1}{2}$
 & $\frac{1}{2}$
 & $\dots$
 & $\frac{1}{2}$
 & $\frac{1}{2}$\\
\cline{2-7}
\multicolumn{3}{c}{}
\end{tabular}
}}
\hspace{2em}
\resizebox{\textwidth}{!}{
\subfigure[Payoff table for ``minority-seeking'' player $i$]{
\begin{tabular}{ r | c | c | c | c | c | c |}
\multicolumn{1}{c}{$x$}
 & \multicolumn{1}{c}{$0$}
 & \multicolumn{1}{c}{$1$}
 & \multicolumn{1}{c}{$2$}
 & \multicolumn{1}{c}{$\dots$}
 & \multicolumn{1}{c}{$n-2$}
 & \multicolumn{1}{c}{$n - 1$}\\
 \cline{2-7}
$u^i_1(x)$
 & $\frac{1}{2}+\left(\frac{1}{2}-\frac{1}{2n}\right)$
 & $\frac{1}{2}+\left(\frac{1}{2}-\frac{3}{2n}\right)$
 & $\frac{1}{2}+\left(\frac{1}{2}-\frac{5}{2n}\right)$
 & $\dots$
 & $\frac{1}{2}-\left(\frac{1}{2}-\frac{3}{2n}\right)$
 & $\frac{1}{2}-\left(\frac{1}{2}-\frac{1}{2n}\right)$\\
\cline{2-7}
$u^i_2(x)$
 & $\frac{1}{2}$
 & $\frac{1}{2}$
 & $\frac{1}{2}$
 & $\dots$
 & $\frac{1}{2}$
 & $\frac{1}{2}$\\
\cline{2-7}
\multicolumn{3}{c}{}
\end{tabular}
}}
\caption{Majority-minority game $G$'s payoffs.
There are $\frac{n}{2}$ majority-seeking players and $\frac{n}{2}$ minority-seeking
players. $x$ denotes the number of players other than $i$ who play 1.
The key feature of the payoff function is that for a majority-seeking
player, the advantage of playing 2 as opposed to 1 decreases linearly in $x$,
and increases linearly in $x$ for a minority-seeking player.}
\label{fig:majority_minority_payoff}
\end{figure}
\end{example}

The proof of \refTheorem{thm:exactNE_lowerbound}
shows that in {\em any} NE of $G$, at least $n/2$ players must use
mixed strategies. Consequently the distribution of the number of players
using either strategy has support $\geq n/2$, so for a typical player
it is necessary to check $n/2$ of his payoffs.

\begin{theorem}\label{thm:exactNE_lowerbound}
A single-payoff query algorithm may need to query
$\Omega(n^2)$ payoffs in order to find an exact Nash equilibrium of an
$n$-player game, even for two-strategy self-anonymous games.
\end{theorem}

\begin{proof}
We prove that in any exact equilibrium of $G$ (Example~\ref{ex:majmin}), at
least $n/2$ players must use mixed strategies.
$G$ is not self-anonymous, but we can then invoke Theorem~\ref{thm:reduction_selfanon},
which reduces the problem of computing equilibria of
anonymous games, to the computation of equilibria of self-anonymous games.

We proceed by arguing that in any Nash equilibrium of $G$, at least $n/2$ players must use
mixed strategies. The result follows, since an equilibrium must then have
support linear in $n$, and for a player who mixes, it becomes
necessary to check, via queries, his payoffs for all these outcomes.
We note in passing that $G$ has an equilibrium in which {\em all} players mix
with equal probabilities $\frac{1}{2}$.

Consider a NE of $G$ and let $p_i$ be the probability that player
$i$ plays 1, in that equilibrium.

Let $P^i_1(s)$ denote the expected payoff of player $i$ for playing strategy 1,
minus the expected payoff of $i$ for playing $2$, in the (possibly mixed) strategy profile
$s$. If, in a pure profile, $x$ other players play strategy $1$, then for a majority
player $i$, $P^i_1(s)=(\frac{x}{n}-\frac{1}{2}+\frac{1}{2n})$.
By linearity of expectations, if in a mixed strategy $s$, $x$ is the
expected number of other players who play 1, then
$P^i_1(s)=(\frac{x}{n}-\frac{1}{2}+\frac{1}{2n})$. Consequently, the incentive for a majority player $i'$ to
play 1 is $\sum_{i\not= i'} \frac{p_i}{n} - \frac{n-1}{2n}$.
Furthermore, the incentive for a minority player $i'$ to play 2 is
$\sum_{i\not= i'} \frac{p_i}{n} - \frac{n-1}{2n}$.

Suppose a majority player $i'$ mixes with probability
$p_{i'}\in(0,1)$.  Notice that $\sum_{i\not= i'} p_i = \frac{n-1}{2}$.
The expected number of users of strategy 1 differs from the expected
number of users of strategy 2 by less than 1. This means that no
majority player may use a pure strategy; if he did, he would have an
incentive to use the opposite strategy.
It follows that all majority players must use mixed strategies.

Suppose a minority player $i'$ who plays a mixed strategy uses $p_{i'} \in (0,1)$.
Suppose, in addition, all majority players play pure strategies.
In that case, as before, any majority player would want to switch.
So in this case, all majority players must mix, as before.

Finally, suppose all players play pure strategies.  If strategies 1
and 2 both have the same number of users, then all majority players
will want to switch.
Alternatively, if, say, strategy 1 is used by more than $n/2$
players, it will be being used by a minority player who will want to
switch. Thus, all majority players must use mixed strategies.

This means that (for each majority player) at least $n/2$ payoffs
need to be known for the computation of expected utilities.
If one of these payoffs is unknown, there remains the possibility that it
could take an alternative value that results in that player no longer
being indifferent between $1$ and $2$, and having an incentive to deviate.
\end{proof}

\subsection{A game whose solution must have irrational numbers}

\cite{dp14} note as an open problem, the question of whether
there is a two-strategy anonymous game whose Nash equilibria
require players to mix with irrational probabilities.
The following example shows that such a game does indeed exist,
even with just 3 players. In the context of the present paper, it is a further
motivation for our focus on approximate rather than exact Nash equilibria.

\begin{example}
\label{ex:irrational}
Consider the following anonymous game represented in normal form in \refFigure{fig:irrational_game_normal_form} and in its anonymous form in \refFigure{fig:irrational_game_anonymous_form}.
We show that in the unique equilibrium, the row, the column, and the matrix
players must randomize respectively with probabilities
\[
p_r = \frac{1}{12}(\sqrt{241} - 7), \ p_c = \frac{1}{16}(\sqrt{241} - 7),
\ p_m = \frac{1}{36}(23 - \sqrt{241}).
\]

\begin{figure}[h!]
\centering
\begin{tabular}{ r | c | c | }
\multicolumn{1}{r}{}
 &  \multicolumn{1}{c}{$1$}
 & \multicolumn{1}{c}{$2$} \\
\cline{2-3}
$1$ \ & $(1, 0, 1)$ & $(1, \frac{1}{2}, 0)$ \\
\cline{2-3}
$2$ \ & $(0, 0, 0)$ & $(\frac{1}{2}, \frac{1}{4}, 0)$ \\
\cline{2-3}
\multicolumn{3}{c}{\qquad$1$}
\end{tabular}
\hspace{2em}
\begin{tabular}{ r | c | c | }
\multicolumn{1}{r}{}
 &  \multicolumn{1}{c}{$1$}
 & \multicolumn{1}{c}{$2$} \\
\cline{2-3}
$1$ \ & $(1, 0, 0)$ & $(0, \frac{1}{4}, \frac{1}{2})$ \\
\cline{2-3}
$2$ \ & $(\frac{1}{2}, 1, \frac{1}{2})$ & $(1, 0, 1)$ \\
\cline{2-3}
\multicolumn{3}{c}{\qquad$2$}
\end{tabular}
\caption{The three-player two-strategy anonymous game in normal
  form. A payoff tuple $(a, b, c)$ represents the row, the
  column, and the matrix players' payoff, respectively.}
\label{fig:irrational_game_normal_form}
\end{figure}
\begin{figure}[h!]
\centering
\subfigure[$r$'s payoff table]{
\begin{tabular}{ r | c | c | c |}
\multicolumn{1}{c}{$x$}
 &  \multicolumn{1}{c}{$0$}
 & \multicolumn{1}{c}{$1$}
 & \multicolumn{1}{c}{$2$}\\
\cline{2-4}
$u^r_1(x)$ & $0$ & $1$ & $1$\\
\cline{2-4}
$u^r_2(x)$ & $1$ & $\frac{1}{2}$ & $0$\\
\cline{2-4}
\multicolumn{3}{c}{}
\end{tabular}
}
\subfigure[$c$'s payoff table]{
\begin{tabular}{ r | c | c | c |}
\multicolumn{1}{c}{$x$}
 &  \multicolumn{1}{c}{$0$}
 & \multicolumn{1}{c}{$1$}
 & \multicolumn{1}{c}{$2$}\\
\cline{2-4}
$u^c_1(x)$ & $1$ & $0$ & $0$\\
\cline{2-4}
$u^c_2(x)$ & $0$ & $\frac{1}{4}$ & $\frac{1}{2}$\\
\cline{2-4}
\multicolumn{3}{c}{}
\end{tabular}
}
\subfigure[$m$'s payoff table]{
\begin{tabular}{ r | c | c | c |}
\multicolumn{1}{c}{$x$}
 &  \multicolumn{1}{c}{$0$}
 & \multicolumn{1}{c}{$1$}
 & \multicolumn{1}{c}{$2$}\\
\cline{2-4}
$u^m_1(x)$ & $0$ & $0$ & $1$\\
\cline{2-4}
$u^m_2(x)$ & $1$ & $\frac{1}{2}$ & $0$\\
\cline{2-4}
\multicolumn{3}{c}{}
\end{tabular}
}
\caption{The three-player two-strategy anonymous game
represented in the anonymous compact form, where $x$ denotes
the number of other players playing strategy $1$.}
\label{fig:irrational_game_anonymous_form}
\end{figure}
\end{example}

\begin{theorem}\label{thm:irrational_nash}
There exists a three-player, two-strategy anonymous game that
has a unique Nash equilibrium where all the players must
randomize with irrational probabilities.
\end{theorem}

\begin{proof}
We use the game in \refExample{ex:irrational}, which is represented
in anonymous form in \refFigure{fig:irrational_game_anonymous_form}.
It is easy to check that the game admits no
pure Nash equilibrium. Let $r, c, m$ denote the row,
column, and matrix player, respectively. Further, let $p_i$ denote
the amount of probability that $i \in \{r, c, m\}$ allocates to strategy $1$.
Suppose, for the moment, that the game admits only fully-mixed
equilibria. Then, these can be found by solving the
following system of equations, which results from making
everyone indifferent.
\[
\begin{cases}
\frac{1}{2} \cdot (p_c \cdot (1 - p_m) + p_m \cdot (1 - p_c)) + p_c \cdot p_m = (1 - p_m) \cdot (1 - p_c)\\
(1 - p_r) \cdot (1 - p_m) = \frac{1}{4} \cdot (p_r \cdot (1 - p_m) + p_m \cdot (1 - p_r)) + \frac{1}{2} \cdot p_r \cdot p_m\\
p_r \cdot p_c = (1 - p_r) \cdot (1 - p_c) + \frac{1}{2} \cdot (p_r \cdot (1 - p_c) + p_c \cdot (1 - p_r)),
\end{cases}
\]
which reduces to
\[
\begin{cases}
\frac{3}{2} \cdot (p_c + p_m) - p_c \cdot p_m = 1\\
\frac{5}{4} \cdot (p_r + p_m) - p_r \cdot p_m = 1\\
\frac{1}{2} \cdot (p_r + p_c) + p_r \cdot p_c = 1,
\end{cases}
\]
and whose unique solution in the interval $[0, 1]$ is
\begin{gather*}
p_r = \frac{1}{12}(\sqrt{241} - 7) \approx 0.71, \ p_c = \frac{1}{16}(\sqrt{241} - 7) \approx 0.53,\\
p_m = \frac{1}{36}(23 - \sqrt{241}) \approx 0.21.
\end{gather*}

Now we show that everybody must indeed randomize in order to be in
equilibrium. Suppose we fix player $r$ to play $1$. Given this, it is
easy to see that both $c$ and $m$ must play $2$, making $r$ unhappy
and willing to move to strategy $2$. If we fix $r$ to play $2$, then
$m$ plays $2$ and $c$ plays $1$. However, $r$ would be better off
deviating to strategy $1$. Similar arguments can prove that we cannot
fix $c$ to any of the two strategies nor $m$ to play $1$. The most
interesting case is when we fix $m$ to play $2$. Given this, we must
set $p_r = \frac{4}{5}$ and $p_c = \frac{2}{3}$ in order for $r$ and
$c$ to be in equilibrium between each other. Player $m$'s expected
payoff for playing $1$ is $\frac{2}{3} \cdot \frac{4}{5} =
\frac{8}{15}$, which is larger than what she gets for playing $2$,
i.e., $\left(\frac{1}{3} \cdot \frac{1}{5}\right) + \frac{1}{2} \cdot
\left(\frac{4}{15} + \frac{2}{15}\right) = \frac{4}{15}$. Hence, $m$ cannot
play pure 2, and the game has a unique fully-mixed Nash equilibrium
where all players must randomize with irrational probabilities.
\end{proof}

\section{Pure equilibria of symmetric games}
\label{section:symmetric_games}

The query complexity lower bound of \refSection{section:exact_nash_equilibria}
led us to investigate whether we can find a NE in the class of
symmetric games, using fewer queries. We show that this is indeed the
case in the context of two-strategy symmetric games and $k$-strategy
self-symmetric games, both of which always possess pure Nash
equilibria (PNE). Since in symmetric games every player shares the same utility function, the following results are w.r.t. single-payoff queries.

\begin{proposition}
\label{prop:upperbound_pne_2-strat}
A Nash equilibrium of any $2$-strategy $n$-player symmetric
game can be found with $O(\log{n})$ single-payoff queries.
\end{proposition}

\begin{algorithm}[h!]
	\caption{SymmetricPNE}
	\label{algorithm:symm_pne}
	\DontPrintSemicolon
	\SetKwFunction{procedure}{search}
	\KwData{The number of players $n$.}
	\KwResult{The number of players $m$ playing strategy $1$ in a PNE.}
	\Begin{
	  	\KwRet \procedure{$0, n - 1$}\;
 	}
	\SetKwProg{myproc}{Procedure}{}{}
	\myproc{\procedure{$\alpha, \beta$}}{
		$m := \lfloor \frac{\alpha + \beta}{2} \rfloor$\;
		\If{$m = \alpha \vee m = \beta$}{
			\KwRet $m$\;
		}
		Use queries to identify: $u_1(m - 1) $, $u_2(m - 1)$, $u_1(m)$, $u_2(m)$\;
		\If{$u_1(m-1)\geq u_2(m-1)~{\rm and}~u_1(m)\leq u_2(m)$}{
			\KwRet $m$\;
		}
		\eIf{$u_1(m-1)<u_2(m-1)$}{
			$\beta := m$\;
		}{
			$\alpha := m$\;
		}
	  	\KwRet \procedure{$\alpha, \beta$}\;
 	}

\end{algorithm} 

\begin{proof}
\refAlgorithm{algorithm:symm_pne} uses a binary search approach to find a
pure Nash equilibrium (PNE).
We know that a PNE must exist, and we show that the algorithm correctly finds one.
Note that when the
conditions are met for returning an output, we have found a PNE.
We show inductively that if a PNE is in the search space of the $k$-th round and
not yet found, then there is a PNE in the search space of round $k+1$.

The base case is trivial since the search space is $\{0, \dots,n-1\}$.
Suppose that after $k$ recursive calls a PNE is still in the search space but not found
yet. We need to show that there is still a PNE in the search space of
step $k + 1$. Let $\{\alpha_k, \dots, \beta_k\}$ be the search space
at step $k$. Let $m_k := \lfloor \frac{\alpha_k + \beta_k}{2} \rfloor$
as in the algorithm. Since $m_k$ does not lead to an equilibrium, the
algorithm makes a case distinction. Note that, by construction,
$u_1(\alpha_k) \geq u_2(\alpha_k)$ and $u_1(\beta_k - 1) < u_2(\beta_k
- 1)$. In fact, in the case $u_1(m_k-1)<u_2(m_k-1)$, we have that the
search space is $\{\alpha_k, \dots, m_k\}$. Due to the induction
hypothesis, a PNE is located at an $x \in \{\alpha_k, \dots,
\beta_k\}$. Hence we only need to show that a PNE is located at an
$x \leq m_k - 1$. A PNE must be there since otherwise it must always
holds that $u_1(x) < u_2(x)$ for all $x \in \{\alpha_k, \dots, m_k\}$,
contradicting the fact that, by construction, $u_1(\alpha_k) \geq
u_2(\alpha_k)$. The other case is symmetric. If $u_1(m_k)\geq
u_2(m_k)$, then the resulting search space is $\{m_k, \dots,
\beta_k\}$. No PNE would require having $u_1(x) \geq u_2(x)$ for all
$x \in \{m_k, \dots, \beta_k\}$. We know, however, that $u_1(\beta_k -
1) < u_2(\beta_k - 1)$, which is a contradiction.

Finally, note that the number of recursive calls made by
\refAlgorithm{algorithm:symm_pne} is $\Theta(\log n)$, with a constant
number of queries per recursive call, giving us the claimed
logarithmic query complexity.
\end{proof}

The following lower bound is given on the restricted class of
self-symmetric games and makes the above algorithm asymptotically
optimal. We use $u :
\Pi_{n}^k \longrightarrow [0, 1]$ to
denote the utility function of any player $i$ due to its
independence from $i$.

\begin{proposition}
\label{prop:lowerbound_pne_2-strat}
Any algorithm that finds a pure Nash equilibrium in a two-strategy, $n$-player self-symmetric
game needs to make $\Omega(\log{n})$ single-payoff queries in the worst case.
\end{proposition}

\begin{proof}
A two-strategy game that is both symmetric and self-anonymous is defined
in terms of a utility function $u:\{0,\ldots,n\}\rightarrow[0,1]$ mapping the
number of players who use strategy 1, to a single payoff that all
players receive. A pure Nash equilibrium (PNE) corresponds to
a local optimum of $u$.

We restrict ourselves to functions $u$ having a unique local optimum $OPT_u$; thus $u(x)$
increases monotonically and then decreases monotonically as $x$ increases.
Assume also that $u(0)=u(n)=0$.
We consider an adversary who answers queries so as to achieve the
following objective. At any point, there is a range that is known to contain $OPT_u$;
each answer to a query should reduce that range by a factor of at most 2.

This is achieved as follows. Let $\{x_1,\ldots,x_2\}$ be the solution range, and
$x_q$ be a query point. If $x_1<x_q<x_2$,
then output a payoff that is slightly higher than previously-observed
payoffs, and let the new solution range be the larger of $\{x_1,\ldots,x_q\}$
and $\{x_q,\ldots,x_2\}$. If $x_q\leq x_1$ or $x_q\geq x_2$, choose a payoff
that interpolates between the closest known values to $x_q$.

With this approach, the adversary can ensure $\Omega(\log(n))$ queries
need to be made, as claimed.
\end{proof}

Moving to $k$-strategy self-symmetric games, these are a special case
of pure coordination games where every pure-strategy profile yields the
same utility to all players \cite{bfh09}. So, once again such a game
possesses a pure Nash equilibrium corresponding to a local
maximum of the utility function $u$.
The following results identify the deterministic and randomized payoff query
complexity of finding a PNE in $k$-strategy self-symmetric games, for any constant $k$.
This is done by relating it to the problem of finding a local optimum of
a real-valued function $f$ on a regular grid graph (Lemma~\ref{lem:grid}), about which
much is known, and then applying those pre-existing results.
Of interest to us is the $d$-dimensional regular grid graph;
$[n]^d$ denotes the graph on $[n]\times\ldots\times[n]$, where edges are only present
between two vertices that differ by 1 in a single coordinate, and other coordinates
are equal.

\begin{lemma}\label{lem:grid}
For any constant $k$, the query complexity of searching for a pure Nash
equilibrium of $k$-strategy $n$-player self-symmetric games, is within a
constant factor of the query complexity of searching for a local optimum
of the grid graph $[n]^{k-1}$.
\end{lemma}

\begin{proof}
Given a self-symmetric game $G$, let ${\cal G}(G)$
be the graph  of pure-strategy profiles of $G$,
whose edges are single-player deviations.
We seek a vertex of ${\cal G}(G)$ which forms a local optimum of $u$.

We start by reducing the problem of searching for an optimum of
a function $f$ of the vertices of $[n]^{k-1}$ to the search for a
PNE of a $k$-strategy self-symmetric game $G$.
First, extend $f$ to the domain $(\{0\}\cup[n])^{k-1}$ in such a way
that no local optimum of $f$ occurs at a point with a zero coordinate.
Assume that the codomain of $f$ is $[0,1]$.
Construct $G$ with $kn$ players and $k$ strategies as follows.
Let $s=(x_1,\ldots,x_k)$ be a typical anonymous pure profile of $G$,
where $x_j$ denotes the number of players who play strategy $j$.
If $0\leq x_j\leq n$ for each $j<k$, set $u(s)$ equal to $f(p)$, where
$p$ is the corresponding point in $[n]^{k-1}$.
For the remaining profiles
(the red area of Figure~\ref{fig:2-simplex_pne}, depicting the case $k=2$),
we assign them negative utilities in such
a way that none of them are local optima of $u$.
Consequently, any PNE of $G$ corresponds to a local optimum of $f$,
so the query complexity of $G$ is at least the query complexity of
finding a local optimum of $f$.
Note that the description size of $G$ is larger than that of $f$ by
just a linear factor (the value of the $k$-th entry of a pure profile of
$G$ is determined by the previous $k-1$ entries).

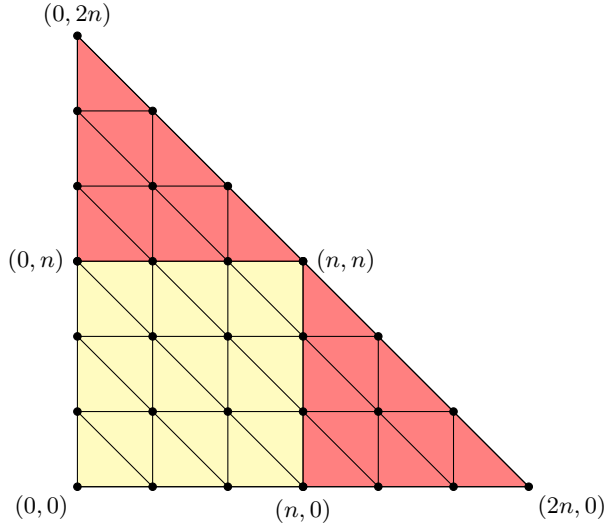
\begin{figure}[h!]
\centering
\begin{tikzpicture}
	\filldraw[fill=yellow, fill opacity=0.25] (0, 0) -- (3, 0) -- (3, 3) -- (0, 3);
	\filldraw[fill=red, fill opacity=0.5] (3, 0) -- (6, 0) -- (3, 3);
	\filldraw[fill=red, fill opacity=0.5] (0, 3) -- (0, 6) -- (3, 3);

	\draw node[below left] {$(0, 0)$} (0,0) -- (6,0) node[below right] {$(2n, 0)$};
	\draw (0, 0) -- (0, 6) node[above] {$(0, 2n)$} (0, 6) -- (6, 0);
	
	\foreach \x in {0, ..., 6} {
		\foreach \y in {\x, ..., 6}
		\node[draw,circle,inner sep=1pt,fill] at (\x, \y - \x) {};
	}
	
	\node[draw,circle,inner sep=1pt,fill, label=left:{$(0, n)$}] at (0, 3) {};
	\node[draw,circle,inner sep=1pt,fill, label=right:{$(n,n)$}] at (3, 3) {};
	\node[draw,circle,inner sep=1pt,fill, label=below:{$(n, 0)$}] at (3, 0) {};
		
	\foreach \x in {1, ..., 5} {
		\draw (\x, 0) -- (0, \x);
		\draw (\x, 0) -- (\x, 6 - \x);
		\draw (0, \x) -- (6 - \x, \x);
	}
\end{tikzpicture}

\caption{The graph when $k = 3$. The yellow area shows the embedded
  grid. The irrelevant points for the local maxima computation are
  in red.}
\label{fig:2-simplex_pne}
\end{figure}

Next we reduce the query-based search for a PNE of a $n$-player $k$-strategy $G$ to
the problem of locally optimizing a function $f$ on a $k$-dimensional grid graph.
We use the graph $[n]^{k-1}$ where the first $k-1$ coordinates of a profile
in $G$ are mapped to the corresponding vertex of $[n]^{k-1}$, and
$f$ maps the remaining vertices of $[n]^{k-1}$ to negative values in such
a way that none of them are solutions.
Reducing in this direction, we have to deal with
non-grid edges of ${\cal G}(G)$ whose presence may cause some profile not to
be a PNE (the diagonal edges in Figure~\ref{fig:2-simplex_pne}).
Such edges correspond to a player switching from a strategy in
$\{1,\ldots,k-1\}$ to another strategy in $\{1,\ldots,k-1\}$ (as opposed
to switching to/from strategy $k$, which correspond to the grid edges).
We deal with this by designing the function $f$ in terms of $u$ such that
each query of the game is simulated by a constant number of queries of
the grid $[n]^{k-1}$. If a vertex $v$ of ${\cal G}(G)$ belongs to one or more
triangles $\{v,v',v''\}$ for which $\{v',v''\}$ is a non-grid edge, then if $u(v)<u(v')$
and $u(v)<u(v'')$, $f(v)$ is set equal to $(u(v')+u(v''))/2$, and if $v$ belongs
to multiple such triangles, the largest such value is used.
Otherwise, $f$ is set equal to $u$.

This construction of $f$ from $u$ gets rid of local optima that may arise from
deleting the diagonal edges (for the purpose of embedding in $[n]^{k-1}$),
and remaining local optima of $f$ correspond to local optima of $u$.
For constant $k$, a query to $f$ can be computed by a constant number
of queries to $u$, so the query complexities (as functions of $n$) are
linearly related.
\end{proof}

\begin{corollary}
The randomized query complexity of searching for PNE of self-symmetric
games is $\Theta(n^{(k-1)/2})$ for constant $k\geq5$.
\end{corollary}

\begin{proof}
We use Lemma~\ref{lem:grid} and a result of \cite{Zhang09} that the randomized query
complexity of optimizing a function on the grid $[n]^d$ is $\Theta(n^{d/2})$ for $d\geq 4$. The statement
follows from the fact that we have $d = k-1$.
\end{proof}

\begin{corollary}\label{prop:lowerbound_pne_k-strat}
The deterministic query complexity of searching for PNE of self-symmetric
games is $\Theta(n^{k-2})$ for constant $k>2$.
\end{corollary}

\begin{proof}
We use Lemma~\ref{lem:grid} and Theorem 3 of \cite{althofer93},
which identifies upper and lower bounds for the query complexity of deterministically
optimizing a function on the grid $[n]^d$ that are both proportional to $n^{d-1}$. Again, the statement
follows from having $d=k-1$. (We note that in \cite{althofer93}, the roles of $n$ and $k$ are reversed relative to here.)
\end{proof}

\section{Self-anonymous Games}
\label{section:self-anonymous_games}

We show that in self-anonymous games, if every player uses the uniform
distribution over her actions, this mixed profile constitutes a $O\left(1 / \sqrt{n}\right)$-WSNE.
This means that no query is needed to find any such approximation.
We provide an inductive proof on the constant number of actions $k$.
First, we demonstrate that this holds for two-strategy games and
subsequently utilize this result as the base case of the induction.
We prove the following lemma below although, in fact,
it also follows from the combination of
\refLemma{lemma:tv_utility} and \refLemma{lemma:tvd_lipschitz}
with $\epsilon = 1/2$.

\begin{lemma}
\label{lemma:two_strategy_sqrt_nash}
In any two-strategy $n$-player self-anonymous game, the mixed-strategy
profile $s = (\frac{1}{2}, \dots, \frac{1}{2})$ is an $O\left(1 /
\sqrt{n}\right)$-WSNE.
\end{lemma}

\begin{proof}
We will show that for any player $i \in [n]$, we have
\begin{equation}\label{condition:strat1}
\Bigl| \mathbb{E}[u^i_1(X_{-i})] - \mathbb{E}[u^i_2(X_{-i})] \Bigr|
   \leq \frac{e}{\pi} \cdot \frac{1}{\sqrt{n - 1}}.
\end{equation}
To prove \refEquation{condition:strat1}, we analyse the expected
value to player 1 of strategy 1 minus that of strategy 2 as follows:
\begin{equation*}
\begin{split}
   & \sum_{x = 0}^{n - 1}(u_1^i(x) - u_2^i(x))
      \cdot \text{Pr}[X_{-i} = x] \\
 = & \sum_{x = 0}^{n - 2}(u_2^i(x + 1) - u_2^i(x))
    \cdot \text{Pr}[X_{-i} = x] + (u_1^i(n - 1) - u_2^i(n - 1))
       \cdot \text{Pr}[X_{-i} = n - 1] \\
 = & \sum_{x = 0}^{n - 2}(u_2^i(x + 1) - u_2^i(x))
    \cdot \binom{n - 1}{x} \cdot \frac{1}{2^{n - 1}} + (u_1^i(n - 1) - u_2^i(n - 1))
       \cdot \frac{1}{2^{n - 1}} \\
 = & \frac{1}{2^{n - 1}}\left(\sum_{x = 1}^{n - 1} u_2^i(x)
    \cdot \left(\binom{n - 1}{x - 1} - \binom{n - 1}{x}\right)
       + (u_1^i(n - 1) - u_2^i(0))\right),
\end{split}
\end{equation*}
where in the second step we applied the self-anonymity property, in
the third step we used the definition of the p.m.f. of the binomial
distribution, and in the fourth one we simply rearranged terms. Since
the utility function outputs values in $[0, 1]$, then $u_1^i(n - 1) -
u_2^i(0) \leq 1$. Moreover, for all $x = 1, \dots, \frac{n - 1}{2}$,
we have $\binom{n - 1}{x - 1} - \binom{n - 1}{x} < 0$, and strictly
positive for the remaining values. Thus, in the worst case we have
\begin{equation*}
u_2^i(x) = 
  \left\{
    \begin{array}{ll}
      0  & \mbox{if } x \in \{1, \dots, \frac{n - 1}{2}\} \\
      1 & \mbox{if } x \in \{\frac{n - 1}{2} + 1, \dots, n - 1\},
    \end{array}
  \right.
\end{equation*}
reducing the above expression to be at most
\begin{equation*}
\begin{split}
 & \frac{1}{2^{n - 1}} \left(\sum_{x = \frac{n - 1}{2} + 1}^{n - 1}
 \left(\binom{n - 1}{x - 1} - \binom{n - 1}{x}\right) + 1\right) \\
 = & \frac{1}{2^{n - 1}}\left(\sum_{x = \frac{n - 1}{2} + 1}^{n - 1}
    \binom{n - 1}{x - 1} - \sum_{x = \frac{n - 1}{2} + 1}^{n - 2} \binom{n - 1}{x}\right)
\end{split}
\end{equation*}
\[
= \frac{1}{2^{n - 1}}\left(\sum_{x = \frac{n - 1}{2}}^{n - 2}
   \binom{n - 1}{x} - \sum_{x = \frac{n - 1}{2} + 1}^{n - 2}
      \binom{n - 1}{x}\right)
       = \frac{1}{2^{n - 1}} \binom{n - 1}{\frac{n - 1}{2}}
  = \frac{1}{2^{n - 1}} \cdot \frac{(n - 1)!}{\left(\left(\frac{n - 1}{2}\right)!\right)^2}.
\]
By Stirling's bounds, we know that
$\sqrt{2\pi} \cdot n^{n + 1/2} \cdot e^{-n} \leq n! \leq e \cdot n^{n + 1/2} \cdot e^{-n}$. Hence,
\begin{gather*}
\frac{1}{2^{n - 1}} \cdot \frac{(n - 1)!}{\left(\left(\frac{n - 1}{2}\right)!\right)^2} \leq
\frac{1}{2^{n - 1}} \cdot \frac{e \cdot (n - 1)^{n - 1 + \frac{1}{2}} \cdot e^{-(n - 1)}}{\left(\sqrt{2\pi} \cdot (\frac{n - 1}{2})^{\frac{n - 1 + 1}{2}} \cdot e^{-(\frac{n-1}{2})}\right)^2} \\
= \frac{1}{2^{n - 1}} \cdot \frac{e \cdot (n - 1)^{n - \frac{1}{2}} \cdot e^{-(n - 1)}}{2\pi \cdot (\frac{n - 1}{2})^n \cdot e^{-(n - 1)}} =
\frac{e}{\pi} \cdot \frac{1}{\sqrt{n - 1}}.
\end{gather*}
This gives the required upper bounded on the value of strategy 1
minus that of strategy 2. The value of strategy 2 minus that of
strategy 1 has the same analysis and upper bound, which gives
us \refEquation{condition:strat1}.
\end{proof}

\begin{theorem}\label{thm:approx_selfanon}
For constant $k$, in any $k$-strategy $n$-player self-anonymous game
letting every player randomize uniformly is a $O\left(1 / \sqrt{n}\right)$-WSNE.
\end{theorem}

\begin{proof}
We proceed by induction on the number of strategies $k$ where
the base case $k = 2$ follows by
\refLemma{lemma:two_strategy_sqrt_nash}. Suppose that
it holds that in any $(k - 1)$-strategy $n$-player self-anonymous
game, every player $i \in [n]$ mixing uniformly is a $O\left(1 /
\sqrt{n}\right)$-WSNE. We show that this holds also for $k$ strategies.

Let $G_{k}$ be a $k$-strategy self-anonymous game. Moreover, let
$X^{(\ell)}_{i}$ be a random variable indicating whether player $i$
plays strategy $\ell$. Let $X^{(k)}_{-i} := \sum_{j \neq i}
X^{(k)}_j$ denote the number of players other than $i$ playing
strategy $k$ in $G_{k}$. We observe that $\mathbb{E}\left[X^{(k)}_{-i}\right] =
\frac{n - 1}{k}$, so by Chernoff bounds, we have that
\begin{equation}\label{eqn:cb}
\text{Pr}\left[X^{(k)}_{-i} \geq
   \frac{2}{k}\left(n - 1\right)\right] \leq e^{-\frac{n - 1}{3k^2}},
\end{equation}
thus, exponentially small in $n$ for constant $k$.
We bound the difference in player $i$'s utility between
two strategies, say, $1$ and $2$, which is
\begin{gather*}
\sum_{\substack{x_1,\ldots,x_k : \\ x_1+\ldots+x_k=n}}
      (u^i_1(x_1,\dots,x_k) - u^i_2(x_1,\dots,x_k))
        \cdot \text{Pr}\left[X^{(1)}_{-i}=x_1,\dots,X^{(k)}_{-i} = x_k\right],
\end{gather*}
where the utility function $u_j^i$ takes as input the number $x_m$ of
players playing strategy $m$, for all $m = 1,\dots,k$.

Next we decompose this into the
contributions to the expected difference between $i$'s utility
for 1 and 2, arising from the event that $x_k$ players play strategy $k$:
\begin{gather*}
\sum_{x_k = 0}^{n} \text{Pr}\left[X^{(k)}_{-i} = x_k\right]
\sum_{\substack{x_1,\ldots,x_{k-1} : \\ x_1+\ldots+x_{k-1}=n-x_k}} \\
(u^i_1(x_2, \dots, x_k) - u^i_2(x_2, \dots, x_k)) \cdot \text{Pr}\left[X^{(2)}_{-i} = x_2, \dots, X^{(k - 1)}_{-i} = x_{k-1} \big| X^{(k)}_{-i} = x_k\right].
\end{gather*}
Observe that when we sum over $x_2,\ldots,x_{k-1}$ but not $x_k$,
we are dealing with an $(n - x_k)$-player $(k - 1)$-strategy game
$G_{k - 1}$ where all $(n - x_k)$ players are still randomizing
uniformly among the $k - 1$ strategies. To see this, we could
think of fixing the identities of the $x_k$ players playing strategy
$k$ to be $\{n - x_k + 1, \dots, n\}$ but $G_{k - 1}$ is anonymous,
i.e., invariant under permutations of the players. By induction
hypothesis, we can bound the difference in payoff in $G_{k - 1}$ by
$O\left(\frac{1}{\sqrt{n - x_k}}\right)$. We can, therefore, write
the difference in utilities between 1 and 2 as
\begin{equation*}
\begin{split}
 & \sum_{x_k = 0}^{n - 1} \text{Pr}\left[X_{- i}^{(k)} = x_k\right] \cdot O\left(\frac{1}{\sqrt{n - x_k}}\right) \\
\leq & \sum_{x_k = 0}^{\frac{2}{k}(n - 1)} \text{Pr}\left[X_{- i}^{(k)} = x_k\right] \cdot O\left(\frac{1}{\sqrt{n - 2 / k}}\right) + \sum_{x_k = \frac{2}{k}(n - 1) + 1}^{n - 1} e^{-\frac{n - 1}{3k^2}} \cdot 1 \\
\leq & \ O\left(\sqrt{\frac{k}{kn - 2}}\right) + \frac{k - 2}{k}n \cdot e^{-\frac{n - 1}{3k^2}} = O\left(\frac{1}{\sqrt{n}}\right).
\end{split}
\end{equation*}

where the first inequality used Equation~(\ref{eqn:cb}).
We should perhaps note that in summing over big-oh expressions,
the hidden constant in the notation is the same for all terms in
the summation (namely, one that applies for $k-1$ strategy games,
which exists by inductive hypothesis).
As long as $k$ is a constant, we have the above claimed upper bound.
In addition, notice that if $k$ is not a constant, the
result no longer holds since the proof adds a constant factor of $k!$ into the
big-oh notation.
\end{proof}

\begin{corollary}
For any constant $k$, no payoff queries are needed to find a $O(1 /
\sqrt{n})$-approximate equilibrium in $k$-strategy $n$-player
self-anonymous games.
\end{corollary}

The following Theorem~\ref{thm:reduction_selfanon} extends an idea
of Lemma~1 of \cite{bfh09}, which reduces the computation of an exact
equilibrium of a two-strategy anonymous game to a two-strategy self-anonymous game. 
Theorem~\ref{thm:reduction_selfanon} goes further, showing that to some extent
approximations can be preserved.
We apply it to prove \refCorollary{cor:selfanon_fptas},
that the problem of finding an FPTAS in two-strategy anonymous games reduces
to finding an FPTAS in two-strategy self-anonymous games.
Below, we also discuss why Theorem~\ref{thm:reduction_selfanon} does not extend
straightforwardly to more than two strategies.

\begin{theorem}
\label{thm:reduction_selfanon}
Let $G := (n, 2, \{u_j^i\}_{i \in [n], j \in [2]})$ be a two-strategy anonymous game,
and let $\bar{G} := (n, 2, \{\bar{u}_j^i\}_{i \in [n], j \in [2]})$ be the following self-anonymous game.
$\bar{u}^i_2(x+1)$ is set equal to $\bar{u}^i_1(x)$ for all $x$, as required
for $\bar{G}$ to be self-anonymous. Then,
\begin{gather*}
\bar{u}^i_2(0) = \frac{1}{2}, \\
\bar{u}_1^i(x) = \bar{u}_2^i(x) + \frac{u_1^i(x) - u_2^i(x)}{2n} \ \text{ for all } x \in \{0, \dots, n - 1\}, \\
\end{gather*}
A mixed strategy profile $s = (p_1, \dots, p_n)$ is an $\epsilon$-approximate
Nash equilibrium in $G$ if and only if it is an $\frac{\epsilon}{2n}$-approximate
Nash equilibrium in $\bar{G}$.
\end{theorem}

\begin{proof}
First of all, we remark that dividing by $2n$ ensures having all the payoffs in $[0, 1]$. Let $X_j$ denote a Bernoulli random variable indicating whether player $j$ plays strategy $1$, whose expectation is $\mathbb{E}[X_j] := p_j$ for all $j \in [n]$, and let $X_{-i} := \sum_{j \neq i} X_j$.
\begin{description}
\item[$\Longrightarrow$:] We now show that if a strategy profile $s = (p_1, \dots, p_n)$ is an $\epsilon$-approximate Nash equilibrium in $G$, then it is an $\frac{\epsilon}{2n}$-approximate Nash equilibrium in $\bar{G}$. Assume for contradiction that $s$ is not an $\frac{\epsilon}{2n}$-approximate Nash equilibrium in $\bar{G}$. Hence, there must be some player $i \in [n]$ that gains strictly more than $\frac{\epsilon}{2n}$ by deviating to either strategy $1$ or $2$, i.e.,
\[
\exists i \in [n] : p_i \cdot \mathbb{E}[\bar{u}_1^i(X_{-i})] + (1 - p_i) \cdot \mathbb{E}[\bar{u}_2^i(X_{-i})] + \frac{\epsilon}{2n} < \mathbb{E}[\bar{u}_\ell^i(X_{-i})],
\]
for some $\ell \in [2]$. We make a case distinction on $\ell$.
\begin{description}
\item[Case $\ell = 1$:] We have that 
\[
(1 - p_i) \cdot \mathbb{E}[\bar{u}_2^i(X_{-i})] + \frac{\epsilon}{2n} < (1 - p_i) \cdot \mathbb{E}[\bar{u}_1^i(X_{-i})],
\]
which by definition is equivalent to
\[
(1 - p_i) \cdot \sum_{x = 0}^{n - 1} \bar{u}^i_2(x) \cdot \text{Pr}[X_{-i} = x] + \frac{\epsilon}{2n} < (1 - p_i) \cdot \sum_{x = 0}^{n - 1} \bar{u}^i_1(x) \cdot \text{Pr}[X_{-i} = x].
\]
If we apply our payoff transformation rule on the RHS, we have
\begin{gather*}
(1 - p_i) \cdot \sum_{x = 0}^{n - 1} \bar{u}^i_2(x) \cdot \text{Pr}[X_{-i} = x] + \frac{\epsilon}{2n} <\\
(1 - p_i) \cdot \sum_{x = 0}^{n - 1} \Bigl(\bar{u}^i_2(x) + \frac{u_1^i(x) - u_2^i(x)}{2n}\Bigr) \cdot \text{Pr}[X_{-i} = x],
\end{gather*}
which reduces to
\[
(1 - p_i) \cdot \sum_{x = 0}^{n - 1} u^i_2(x) \cdot \text{Pr}[X_{-i} = x]  + \epsilon < (1 - p_i) \cdot \sum_{x = 0}^{n - 1} u^i_1(x) \cdot \text{Pr}[X_{-i} = x].
\]
Applying the definition of expected utility we have
\[
(1 - p_i) \cdot \mathbb{E}[u_2^i(X_{-i})]  + \epsilon < (1 - p_i) \cdot \mathbb{E}[u_1^i(X_{-i})].
\]
Therefore, we contradict the fact that $s$ is an $\epsilon$-equilibrium for $G$ since $i$ prefers deviating to strategy $1$ also in $G$.

\item[Case $\ell = 2$:] We have that
\[
p_i \cdot \mathbb{E}[\bar{u}_1^i(X_{-i})] + \frac{\epsilon}{2n} < p_i \cdot \mathbb{E}[\bar{u}_2^i(X_{-i})],
\]
and, by applying the same arguments, we have 
\[
p_i \cdot \mathbb{E}[u_1^i(X_{-i})]  + \epsilon < p_i \cdot \mathbb{E}[u_2^i(X_{-i})].
\]
Therefore, we contradict the fact that $s$ is an $\epsilon$-equilibrium for $G$ since $i$ prefers deviating to strategy $2$ also in $G$.
\end{description}

\item[$\Longleftarrow$:] We now show that the converse is also true. The proof is very similar to the previous one. We assume for contradiction that $s$ is not an $\epsilon$-equilibrium in $G$. Then,
\[
\exists i \in [n] : p_i \cdot \mathbb{E}[u_1^i(X_{-i})] + (1 - p_i) \cdot \mathbb{E}[u_2^i(X_{-i})] + \epsilon < \mathbb{E}[u_\ell^i(X_{-i})],
\]
for some $\ell \in [2]$. We make, again, a case distinction on $\ell$.
\begin{description}
\item[Case $\ell = 1$:] We have that
\[
(1 - p_i) \cdot \mathbb{E}[u_2^i(X_{-i})] + \epsilon < (1 - p_i) \cdot \mathbb{E}[u_1^i(X_{-i})],
\]
which by definition is equivalent to
\[
(1 - p_i) \cdot \sum_{x = 0}^{n - 1} (u^i_1(x) - u^i_2(x)) \cdot \text{Pr}[X_{-i} = x] > \epsilon.
\]
If we apply our transformation $u_1^i(x) - u_2^i(x) = 2n \cdot (\bar{u}_1^i(x) - \bar{u}_2^i(x))$, we have that
\[
(1 - p_i) \cdot \sum_{x = 0}^{n - 1} (\bar{u}^i_1(x) - \bar{u}^i_2(x)) \cdot \text{Pr}[X_{-i} = x] > \frac{\epsilon}{2n},
\]
which means
\[
(1 - p_i) \cdot \mathbb{E}[\bar{u}_1^i(X_{-i})]  + \frac{\epsilon}{2n} < (1 - p_i) \cdot \mathbb{E}[\bar{u}_2^i(X_{-i})],
\]
contradicting the assumption that $s$ is an $\frac{\epsilon}{2n}$-equilibrium in $\bar{G}$.

\item[Case $\ell = 2$:] Similarly, we have that
\[
p_i \cdot \mathbb{E}[u_1^i(X_{-i})] + \epsilon < p_i \cdot \mathbb{E}[u_2^i(X_{-i})],
\]
which, as in the above case, reduces to having
\[
p_i \cdot \mathbb{E}[\bar{u}_1^i(X_{-i})]  + \frac{\epsilon}{2n} < p_i \cdot \mathbb{E}[\bar{u}_2^i(X_{-i})],
\]
contradicting the assumption that $s$ is an $\frac{\epsilon}{2n}$-equilibrium in $\bar{G}$.
\end{description}
\end{description}
\end{proof}

\begin{corollary}
\label{cor:nash_selfanon}
Let $G, \bar{G}$ be defined as above. A strategy profile $s = (p_1, \dots, p_n)$ is a Nash equilibrium in $G$ if and only if it is a Nash equilibrium in $\bar{G}$.
\end{corollary}

\begin{proof}
Use \refTheorem{thm:reduction_selfanon} with $\epsilon = 0$.
\end{proof}

\begin{corollary}
\label{cor:selfanon_nash}
Finding a Nash equilibrium in a two-strategy anonymous game is as hard as finding a Nash equilibrium in a two-strategy self-anonymous game.
\end{corollary}

\begin{proof}
Given any two-strategy anonymous game $G$, we can certainly construct $\bar{G}$ in time that is linear in the size of $G$, i.e., $O(n^2)$. By \refCorollary{cor:nash_selfanon}, the two games have the same set of equilibria.
\end{proof}

\begin{corollary}
\label{cor:selfanon_fptas}
If there is an FPTAS for two-strategy self-anonymous games, then there is also an FPTAS for two-strategy anonymous games.
\end{corollary}

\begin{proof}
Pick any two-strategy anonymous game $G$, construct $\bar{G}$ and find a $\delta$-equilibrium for $\delta = \frac{\epsilon}{2n}$. Since there is an FPTAS, we have a running time that is poly$(n, 1 / \delta)$. For $\delta = \frac{\epsilon}{2n}$, the running time clearly is poly$(n, 1 / \epsilon)$. \refTheorem{thm:reduction_selfanon} implies that a $\delta$-equilibrium in $\bar{G}$ is an $\epsilon$-equilibrium in $G$.
\end{proof}

We end this section with a brief discussion of why
Theorem~\ref{thm:reduction_selfanon} (and hence Corollaries~\ref{cor:nash_selfanon},\ref{cor:selfanon_nash},\ref{cor:selfanon_fptas}) does not extend
straightforwardly to more than two strategies.
In the context of 3-strategy games, let $u^i_j(x,y,z)$ denote player $i$'s utility
when amongst the $n-1$ remaining players, $x$ play 1, $y$ play 2, and $z$ play 3.
If there were, say, 16 players, note that to be self-anonymous,
the corresponding game $\bar{G}$ would need to satisfy
\[
\bar{u}^i_1(4,6,5)=\bar{u}^i_2(5,5,5)=\bar{u}^i_3(5,6,4)
\]
\[
\bar{u}^i_1(5,5,5)=\bar{u}^i_2(6,4,5)=\bar{u}^i_3(6,5,4)
\]
\[
\bar{u}^i_1(5,6,4)=\bar{u}^i_2(6,5,4)=\bar{u}^i_3(6,6,3).
\]
The proof of  Theorem~\ref{thm:reduction_selfanon} envisages that $\bar{u}^i_1(5,6,4)-\bar{u}^i_3(5,6,4)$ 
should be a rescaled version of $u^i_1(5,6,4)-u^i_3(5,6,4)$, and similarly for
$\bar{u}^i_1(5,5,5)-\bar{u}^i_2(5,5,5)$ and $\bar{u}^i_2(6,5,4)-\bar{u}^i_3(6,5,4)$.
That is, the pairwise differences between the three values given by the
above displayed formulae, are required to be proportionate to pairwise
differences between three pairs of payoffs in $G$, which themselves may take
any values in $[0,1]$. Hence we get a conflict between the constraints that
result from a natural approach to generalising the result.
It may be that (for $k>2$) there exists a more complicated reduction from an FPTAS for $k$-strategy games
to an FPTAS for $k$-strategy self-anonymous games, perhaps by increasing the number of players.

\section{Two-strategy Lipschitz Games}
\label{section:lipschitz_games}

Lipschitz games are anonymous games where every player's utility function
is Lipschitz-continuous, in the sense that for all $i \in [n]$,
all $j \in [k]$, and all $x, y \in \Pi_{n-1}^k$, it holds that
$\left|u^i_j(x) - u^i_j(y)\right| \leq \lambda \left\| x - y \right\|_1$, where $\lambda \geq 0$
is the Lipschitz constant.

For games satisfying a Lipschitz condition with a small value of $\lambda$, e.g. inverse polynomial in $n$,
we obtain positive results for approximation and query complexity.
In the special case of games with two strategies per player, we
show how a solution can be efficiently found via a binary search on $\{0, \dots, n-1\}$.
However, this binary-search approach cannot be used for $k$-strategy games with $k > 2$.

The following algorithms exploit our knowledge of the existence of
{\em pure} approximate equilibria in Lipschitz games \cite{as13, dp14}.
\refAlgorithm{algorithm:approxNELip} is subsequently used in
\refSection{section:two_strategy_anonymous}
as a subroutine of an algorithm for general (not necessarily Lipschitz) anonymous games. For notational convenience with regard to this subsequent result, in \refAlgorithm{algorithm:approxNELip} and \refTheorem{thm:lipschitz2strategy} we use $\bar{G}$ to denote the input game rather than $G$.

\begin{definition}
Let $(j, x) \in \strats \times \{0, \dots, n-1\}$ be the input for
an all-players query. For $\delta\geq 0$, a \emph{$\delta$-accurate}
all-players query returns a tuple of values $(f^1_j(x),\dots,f^n_j(x))$
such that for all $i \in [n]$, $ |u^i_j(x) - f^i_j(x)|\leq\delta$,
i.e., they are within an additive $\delta$ of the correct
payoffs $(u^1_j(x), \dots, u^n_j(x))$.
\end{definition}

\begin{algorithm}[h!]
    \caption{Approximate NE Lipschitz}
    \label{algorithm:approxNELip}
    \DontPrintSemicolon   
    \KwData{$\delta$-accurate query access to utility function $\bar{u}$ of $n$-player
                  $\lambda$-Lipschitz game $\bar{G}$.}
    \KwResult{pure-strategy $3(\delta+\lambda)$-NE of $\bar{G}$.}
    \Begin{
        Let $BR_1(i)$ be the number of players whose best response (as derived
        from the $\delta$-accurate queries) is 1 when $i$ of the
        other players play 1 and $n-1-i$ of the other players play 2.\;
        Define $\phi(i)=BR_1(i)-i$. \tcp*{by construction, $\phi(0)\geq 0$} \tcp*{and $\phi(n-1)\leq 0$}
        If $BR_1(0)=0$, \KwRet all-1's profile.\;
        If $BR_1(n-1)=n$, \KwRet all-2's profile.\;
        Otherwise, \tcp*{In this case, $\phi(0)>0$ and $\phi(n-1)\leq 0$}
        Find, via binary search, $x$ such that $\phi(x)>0$ and $\phi(x+1)\leq 0$.\;
        Construct pure profile $\bar{p}$ as follows:\;
        For each player $i$, if $\bar{u}^i_1(x)-\bar{u}^i_2(x)>2\delta$, let $i$ play 1, and
        if $\bar{u}^i_2(x)-\bar{u}^i_1(x)>2\delta$, let $i$ play 2. (The $\bar{u}^i_j$'s
        are $\delta$-accurate.)
        Remaining players are allocated
        either 1 or 2, subject to the constraint that $x$ or $x+1$ players in total play 1.\\
        \KwRet $\bar{p}$.
        }
\end{algorithm}

\begin{theorem}
\label{thm:lipschitz2strategy}
Let $\bar{G} = \left(n, 2, \{\bar{u}^i_j\}_{i \in [n], j \in \strats}\right)$ be an $n$-player, two-strategy $\lambda$-Lipschitz anonymous game.
\refAlgorithm{algorithm:approxNELip} finds a pure-strategy
$3(\lambda + \delta)$-WSNE using
$4\log{n}$ $\delta$-accurate all-players payoff queries.
\end{theorem}

\begin{proof}
Consider the function $\phi:\{0,\ldots,n-1\}\longrightarrow\{-n,\ldots,n\}$
defined in \refAlgorithm{algorithm:approxNELip}.
It can be readily checked that if
$BR_1(0)=0$ or $BR_1(n-1)=n$, then the solutions returned are correct.
Alternatively, the algorithm has to find $x$ such that $\phi(x)>0$ and $\phi(x+1)\leq 0$,
in the case when $\phi(0)> 0$ and $\phi(n-1)\leq 0$,
and it is clear that $\log{n}$ evaluations of $\phi$ suffice to find $x$.
Moreover, notice that a candidate $x$ can be checked using the following 4
all-players queries: $(1, x), (2, x), (1, x+1), (2, x+1)$. Hence, we need in total
$4\log{n}$ queries to find $x$.

The main task is to prove that given $x$ satisfying these conditions,
a pure profile $\bar{p}$ can indeed be constructed in the way described,
and that $\bar{p}$ is indeed an approximate equilibrium.

Suppose for a contradiction that $\bar{p}$ could not be constructed in the
way described. For example, suppose that more than $x+1$ players are
required to play 1 due to satisfying $\bar{u}^i_1(x)-\bar{u}^i_2(x)>2\delta$.
We argue that this would contradict that $\phi(x+1)\leq 0$.
The fact that $\phi(x+1)\leq 0$ means that there are $n-(x+1)$ players whose
payoffs to play 2 (when $x+1$ others play 1) are at most $2\delta$ less
than their payoffs to play 1 (when $x+1$ others play 1).
When these players play 2, they are $2\delta$-best-responding if $x+1$
players play 1, and by the Lipschitz condition are $2(\lambda + \delta)$-best-responding
if $x$ players play 1. So there is in fact a solution with only $x+1$ players
playing 1. A similar argument rules out the possibility that too many players
are required to play 2.

The fact that $\bar{p}$ is a $3(\lambda+\delta)$-approximate equilibrium
follows from the constraints imposed on which pure strategy is allocated
to each player.
\end{proof}

\section{General Two-strategy Anonymous Games}
\label{section:two_strategy_anonymous}

First, we present our main result (\refTheorem{thm:approximateNE}).
Afterwards, we prove a lower bound on the number of queries that any
randomized algorithm needs to make to find any $\epsilon$-WSNE.

\subsection{Upper Bound}

Before going into technical lemmas, we provide an informal overview of
the algorithmic approach. Suppose we are to solve an $n$-player
game $G$. The first idea is to ``smooth'' every player's utility
function, so that it becomes $\lambda$-Lipschitz continuous for some $\lambda$.
We smooth a utility function by requiring every player to use some amount of randomness.
Specifically, for some small $\zeta$ we make every player place probability either
$\zeta$ or $1 - \zeta$ onto strategy 1.
Consequently, the expected payoff for player $i$ is obtained by averaging her payoff
values w.r.t. a sum of two binomial distributions, consisting of a discrete
bell-shaped distribution whose standard deviation is at least $\zeta\sqrt{n}$.

We construct the smooth game $\bar{G}$ in the following manner.
The payoff received in $\bar{G}$ by player $i$ when $x$ other players are playing
strategy 1 is given by the expected payoff received in $G$ by player $i$ when
$x$ other players play 1 with probability $1-\zeta$ and $n-1-x$ other players
play 1 with probability $\zeta$.
This creates a $\lambda$-Lipschitz game $\bar{G}$ with $\lambda =
O\left(1 / \zeta\sqrt{n}\right)$.

Due to dealing with a two-strategy Lipschitz game, we can use the
bisection method of \refAlgorithm{algorithm:approxNELip}.
If we were allowed to query $\bar{G}$ directly,
a logarithmic number of all-players queries would suffice.
Unfortunately, this is not the case; thus, we need to simulate a query to $\bar{G}$
with a small number of queries to the original game $G$.
Those queries are randomly sampled from the mixed anonymous profile
above, and we take enough samples to ensure we get good estimates
of the payoffs in $\bar{G}$ with sufficiently high probability.

Thus, we are able to find an approximate pure Nash
equilibrium of $\bar{G}$ with $\tilde{O}(\sqrt{n})$ all-players
queries. This equilibrium is mapped
back to $G$ by letting the players who play strategy 1 in $\bar{G}$,
play it with probability $1 - \zeta$ in $G$, and the ones who play strategy 2 in
$\bar{G}$ place probability $\zeta$ on strategy 1 in $G$. The
quality of the approximation is proportional to $\left(\zeta +
(\zeta \sqrt{n})^{-1}\right)$.

Before presenting our main algorithm (\refAlgorithm{algorithm:approximateNE})
and proving its efficiency (\refTheorem{thm:approximateNE}),
we state the following lemmas that are used in the proof.

\begin{lemma}[\cite{dp14}]
\label{lemma:tv_utility}
Let $X, Y$ be two random variables over $\{0,\ldots,n\}$ such that $\left\|X
- Y\right\|_{\TV} \leq \delta$ (where $\left\|X
- Y\right\|_{\TV}$ denotes the total variation distance between $X$
and $Y$, i.e., $1/2 \cdot \sum_{x = 0}^n \left|\Pr[X = x] - \Pr[Y = x]\right|$).
Let $f : \{0, \dots, n\} \longrightarrow [0, 1]$. Then,
\[
\sum_{x = 0}^n f(x) \cdot
(\Pr[X = x] - \Pr[Y = x]) \leq 2 \delta.
\]
\end{lemma}

\begin{lemma}[Simulation of a query to $\bar{G}$ (\refAlgorithm{algorithm:approximateNE})]
\label{lemma:empirical_utility}
Let $\delta, \tau > 0$. Let $X$ be the sum of $n-1$ Bernoulli random variables representing a
mixed profile of $n-1$ players in an $n$-player game $G$. Suppose we want to estimate, with
additive error $\delta$, the expected payoffs
$\expected[u^i_j\left(X\right)]$ for all $i \in [n], j \in \{1,2\}$.
This can be done with probability $\geq 1 - \tau$ using $(1/2\delta^2) \cdot \log{(4n / \tau)}$
all-players queries.
\end{lemma}

\begin{proof}
Suppose we draw a set of $N$ random samples $\{Z_1,\dots,Z_N\}$
from the probability distribution of $X$ (which can be done by computing
each $Z_i$ as a sum of $0/1$ outcomes of $n-1$ biased coin flips.
For each $Z_\ell$ and each $j\in\{1,2\}$ we can make an all-players
query that tells us, for every player, that payoff obtained by that player
for playing $j$ when $Z_i$ other players play 1.
So, a total of $2N$ queries are made.

Let $\hat{U}_j^i := 1/N \cdot \sum_{\ell = 1}^N u^i_j(Z_\ell)$
denote our estimate of $\expected[u^i_j(X)]$. Then
$\expected\left[\hat{U}_j^i\right] = \expected[u^i_j\left(X\right)]$.

We can now use Hoeffding's inequality to get that
\[
\Pr\left[\left|\hat{U}_j^i  - \expected\left[u^i_j\left(X\right)\right] \right| \geq \delta \right]
\leq 2 \text{exp} \left(-2 \delta^2 N\right).
\]

Since there are $2n$ quantities we desire to estimate ($2$ strategies per player)
within additive error $\delta$, we require a failure probability of at most $\tau/2n$, so
that with a union bound we get that the estimates have additive error $\leq\delta$ with probability at least $1-\tau$.
Thus, we need that $ 2 \text{exp} \left(-2 \delta^2 N\right) \leq \tau/ 2n$, which is satisfied for
$N \geq (1/2\delta^2) \cdot \log(4n/\tau)$.
\end{proof}

\begin{lemma}\label{lemma:upper_bound_mode_binomial}
Let $Y$ be a binomial random variable, $Y\sim B(n,p)$ for $p\in[\zeta,1-\zeta]$,
where $0<\zeta <\frac{1}{2}$. Then, the probability value at $Y$'s mode
is at most $\frac{e}{2\pi \zeta \sqrt{n}} \left(1 + \frac{1}{\zeta
  n}\right)$, i.e., $O\left(\frac{1}{\zeta \sqrt{n}}\right)$.
\end{lemma}

\begin{proof}
The mode of $Y$ is either $\lfloor np\rfloor$ or $\lfloor np\rfloor+1$.
We bound the ratio $\Pr[Y = \lfloor np \rfloor + 1]/\Pr[Y=\lfloor np \rfloor]$,
then we apply Stirling's bound on the value $\Pr[Y=\lfloor np \rfloor]$.
We bound the ratio as follows, where we use $\{a\}$ to denote the fractional part of $a$.
\begin{equation*}
\begin{split}
\frac{\binom{n}{\lfloor np \rfloor + 1} \cdot p^{\lfloor np \rfloor + 1} \cdot (1 - p)^{n - \lfloor np \rfloor - 1}}{\binom{n}{\lfloor np \rfloor} \cdot p^{\lfloor np \rfloor} \cdot (1 - p)^{n - \lfloor np \rfloor}} & = \frac{p}{1 - p} \cdot \frac{\lfloor np \rfloor! \cdot (n - \lfloor np \rfloor)!}{(\lfloor np \rfloor + 1)! \cdot (n - \lfloor np \rfloor - 1)!}\\
& = \frac{p}{1 - p} \cdot \frac{n - \lfloor np \rfloor}{\lfloor np \rfloor + 1}\\
& = \frac{p}{1 - p} \cdot \frac{n - np + \{np\}}{\lfloor np \rfloor + 1}\\
& = \frac{p}{1 - p} \cdot \left(\frac{(1 - p)n}{\lfloor np \rfloor + 1} + \frac{\{np\}}{\lfloor np \rfloor + 1}\right)\\
& \leq \frac{np}{np} + \frac{p \cdot \{np\}}{(1 - p) \cdot np}\\
& < 1 + \frac{1}{\zeta n}.
\end{split}
\end{equation*}
In the second to last step we used the fact that $\lfloor np \rfloor +
1 \geq np$. In the last one we used both $\{np\} < 1$ and $1 - p \geq
\zeta$. Next we bound the value at $x = \lfloor np \rfloor$ using Stirling's bounds.
\begin{equation*}
\begin{split}
\Pr[Y = x] & = \binom{n}{x} \cdot p^x \cdot (1 - p)^{n - x}\\
& \leq \frac{e \cdot n^{n + 1/2} \cdot e^{-n} \cdot p^x \cdot (1 - p)^{n - x}}{\sqrt{2\pi} \cdot x^{x + 1/2} \cdot e^{-x} \cdot \sqrt{2\pi} \cdot (n - x)^{n - x + 1/2} \cdot e^{-(n - x)}}\\
& = \frac{e}{2\pi} \cdot \frac{n^{n + 1/2} \cdot p^x \cdot (1 - p)^{n - x}}{(\lfloor np\rfloor)^{x + 1/2} \cdot (n - \lfloor np\rfloor)^{n - x + 1/2}}\\
& = \frac{e}{2\pi} \cdot \frac{n^{n + 1/2} \cdot p^{\lfloor np\rfloor} \cdot (1 - p)^{n - \lfloor np\rfloor}}{n^{\lfloor np\rfloor + 1/2} \cdot n^{n - \lfloor np\rfloor + 1/2} \cdot p^{\lfloor np\rfloor + 1 / 2} \cdot (1 - p)^{n - \lfloor np\rfloor + 1 / 2}}\\
& = \frac{e}{2 \pi} \cdot \frac{1}{\sqrt{p (1 - p)} \cdot \sqrt{n}}\\
& \leq \frac{e}{2 \pi} \cdot \frac{1}{\zeta \sqrt{n}} = O\left(\frac{1}{\zeta\sqrt{n}}\right).
\end{split}
\end{equation*}
If we combine $\Pr[Y = m] / \Pr[Y = x] < 1 +
\frac{1}{\zeta n}$ with $\Pr[Y = x] \leq \frac{e}{2 \pi}
\cdot \frac{1}{\zeta \sqrt{n}}$, it follows that $\Pr[Y = m]
\leq \frac{e}{2\pi \zeta \sqrt{n}} (1 + \frac{1}{\zeta n}) =
O(\frac{1}{\zeta \sqrt{n}})$, concluding the proof.
\end{proof}

\begin{lemma}
\label{lemma:mode_pbd}
Let $Z := \sum_{i = 1}^n Z_i$ be the sum of $n$ independent $0$-$1$
random variables such that for some $0<\zeta\leq 1/2$,
$\expected[Z_i]\in\{\zeta,1-\zeta\}$ for all $i\in[n]$.
Then, the probability value at $Z$'s mode is $O\left(\frac{1}{\zeta\sqrt{n}}\right)$.
\end{lemma}

\begin{proof}
$Z=X+Y$ for binomial random variables $X,Y$ with
$X\sim B(n_X,\zeta)$ and $Y\sim B(n_Y,1-\zeta)$, where $n_X+n_Y=n$.
$Z$'s probability mass function can be written as
\begin{equation*}
\Pr[Z = i] = \sum_{x = 0}^{n_X} \Pr[Z = i | X = x] \cdot \Pr[X = x] = \sum_{x = 0}^{n_X} \Pr[X = x] \cdot \Pr[Y = i - x].
\end{equation*}
Since $n_X+n_Y=n$, we have $\max\{n_X,n_Y\}\geq n/2=\Omega(n)$.
Assume $n_X$ be the maximum. Then, by
\refLemma{lemma:upper_bound_mode_binomial}, Pr$[X = x] \leq
O(\frac{1}{\zeta \sqrt{n_X}})$ for all $x = 0, \dots, n_X$. Hence,
\begin{equation*}
\begin{split}
\Pr[Z = i] & \leq \sum_{x = 0}^{n_X} O\left(\frac{1}{\zeta \sqrt{n_X}}\right) \cdot \Pr[Y = i - x]\\
& = O\left(\frac{1}{\zeta \sqrt{n_X}}\right) \cdot \sum_{x = 0}^{n_X} \Pr[Y = i - x]\\
& \leq O\left(\frac{1}{\zeta \sqrt{n_X}}\right) = O\left(\frac{1}{\zeta\sqrt{n}}\right).
\end{split}
\end{equation*}
\end{proof}

\begin{lemma}
\label{lemma:tvd_lipschitz}
Let $X^{(j,n)} := \sum_{i \in [n]} X_i$ denote the sum of $n$
independent $0/1$ random variables where $\expected[X_i]=1-\zeta$
for all $i\leq j$, and $\expected[X_i]=\zeta$ for all $i>j$.
Then, for all $j\in[n]$, we have
$\left\|X^{(j-1,n)}-X^{(j,n)}\right\|_{\TV}=O\left(\frac{1}{\zeta \sqrt{n}}\right)$.
\end{lemma}

\begin{proof}
We use the following recursive formula for the probability mass
function of a Poisson Binomial Distribution, as described in
\cite{hong13}. Due to this not depending on $j$, we use $X^{(*,
  n)}$ to denote a sum of $n$ independent $0/1$ random variables
whose expectations can potentially be all different.
Then if $p_n$ is the expectation of the $n$-th variable,
\begin{equation}
\label{eq:rf_pbd}
\Pr\left[X^{(*, n)} = i\right] = (1 - p_n) \cdot \Pr\left[X^{(*, n - 1)} = i\right] + p_n \cdot \Pr\left[X^{(*, n - 1)} = i - 1\right].
\end{equation}
We want to bound the total variation distance between $X^{(j - 1, n)}$ and $X^{(j, n)}$, i.e.,
\[
\left\|X^{(j - 1, n)} - X^{(j, n)}\right\|_{\TV} = \frac{1}{2} \sum_{i = 1}^{n} \left|\Pr\left[X^{(j - 1, n)} = i\right] - \Pr\left[X^{(j, n)} = i\right]\right|.
\]
Note that $X^{(j - 1, n)}$ and $X^{(j, n)}$ differ only by how one
coin flip is biased. Thus, we can use $\refEquation{eq:rf_pbd}$ to
write
\begin{gather*}
\Pr\left[X^{(j - 1, n)} = i\right] = (1 - \zeta) \cdot \Pr\left[X^{(j - 1, n - 1)} = i\right] + \zeta \cdot \Pr\left[X^{(j - 1, n - 1)} = i - 1\right], \text{ and}\\
\Pr\left[X^{(j, n)} = i\right] = \zeta \cdot \Pr\left[X^{(j - 1, n - 1)} = i\right] + (1 - \zeta) \cdot \Pr\left[X^{(j - 1, n - 1)} = i - 1\right].
\end{gather*}
If we combine these with the total variation distance
expression, we have the following expression for variation distance:
\begin{gather*}
\frac{1}{2} \sum_{i = 1}^{n - 1} \left|(1 - 2\zeta) \cdot \Pr\left[X^{(j, n - 1)} = i\right] - (1 - 2\zeta) \cdot \Pr\left[X^{(j, n - 1)} = i - 1\right]\right| \\
= \frac{(1 - 2\zeta)}{2} \sum_{i = 1}^{n - 1} \left|\Pr\left[X^{(j, n - 1)} = i\right] - \Pr\left[X^{(j, n - 1)} = i - 1\right]\right|.
\end{gather*}

By definition of $X^{(j, n - 1)}$ and \refLemma{lemma:mode_pbd} we know
that $\Pr[X^{(j - 1, n - 1)} = i] \leq O(\frac{1}{\zeta
  \sqrt{n - 1}})$ for any $i \in \{0, \dots, n - 1\}$.
Let $m$ be the mode of $X^{(j - 1, n - 1)}$.
Due to $X^{(j - 1, n - 1)}$ being unimodal, we can split the sum $\sum_{i = 0}^{n
  - 1} |\Pr[X^{(j - 1, n - 1)} = i] - \Pr[X^{(j - 1, n -
    1)} = i - 1]|$ into two summations over $\{0, \dots, m\}$ and $\{m +
1, \dots, n - 1\}$ where $\Pr[X^{(j-1,n-1)}=i]$ is, respectively, increasing or
decreasing and, hence, remove the absolute value operator.
Consider the summation over $\{0, \dots, m\}$; the other case is
symmetric. Then,
\begin{gather*}
\sum_{i = 1}^{m} \Pr\left[X^{(j - 1, n - 1)} = i\right] - \Pr\left[X^{(j - 1, n - 1)} = i - 1\right]\\
 = \left(\Pr\left[X^{(j - 1, n - 1)} = 1\right] - \Pr\left[X^{(j - 1, n - 1)} = 0\right]\right) + \dots\\
\dots + \left(\Pr\left[X^{(j - 1, n - 1)} = m\right] - \Pr\left[X^{(j - 1, n - 1)} = m - 1\right]\right)\\
 = \Pr\left[X^{(j - 1, n - 1)} = m\right] - \Pr\left[X^{(j - 1, n - 1)} = 0\right] \leq \Pr\left[X^{(j - 1, n - 1)} = m\right].
\end{gather*}
Summing up both the increasing and decreasing side bounds, we get that 
\begin{gather*}
\sum_{i = 0}^{n - 1} \left|\Pr\left[X^{(j - 1, n - 1)} = i\right] - \Pr\left[X^{(j - 1, n - 1)} = i - 1\right]\right| \leq\\
2 \cdot \Pr\left[X^{(j - 1, n - 1)} = m\right] = O\left(\frac{1}{\zeta \sqrt{n - 1}}\right).
\end{gather*}
Substituting back to the variation distance expression and observing
that $1 - 2\zeta < 1$, we obtain
\[
\left\|X^{(j - 1, n)} - X^{(j, n)}\right\|_{\TV} \leq O\left(\frac{1}{\zeta \sqrt{n - 1}}\right) = O\left(\frac{1}{\zeta \sqrt{n}}\right),
\]
as in the statement of the Lemma.
\end{proof}

\begin{definition}
\label{def:smooth_game}
Let $G=(n,2,\{u_j^i\}_{i \in [n], j \in \strats})$ be an anonymous game.
For $\zeta > 0$, the $\zeta$-smoothed version of $G$ is a game
$\bar{G} = (n,2, \{\bar{u}_j^i\}_{i \in [n], j \in \strats})$
defined as follows.
Let $X^{(x)}_{-i} := \sum_{\ell \neq i} X_\ell$ denote the
sum of $n-1$ Bernoulli random variables where $x$ of them have
expectation equal to $1 - \zeta$, and the remaining ones have
expectation equal to $\zeta$.
The payoff $\bar{u}^i_j(x)$ obtained by player
$i \in [n]$ for playing strategy $j \in \strats$ against $x \in \{0,
\dots, n - 1\}$ is
\[
\bar{u}^i_j(x) := \sum_{y = 0}^{n - 1} u^i_j(y) \cdot \Pr\left[X^{(x)}_{-i} = y\right] =
\expected\left[u_j^i\left(X_{-i}^{(x)}\right)\right].
\]
\end{definition}

\begin{theorem}\label{thm:approximateNE}
Let $G=(n,2,\{u_j^i\}_{i \in [n], j \in \strats})$ be an anonymous game.
For $\epsilon$ satisfying $1/\epsilon=O(n^{1/4})$,
\refAlgorithm{algorithm:approximateNE} can be used to find
(with probability $\geq\frac{3}{4}$) an $\epsilon$-NE of $G$,
using $O(\sqrt{n}\cdot \log^2{n})$
all-players queries (hence, $O(n^{3/2} \cdot \log^2{n})$
single-payoff queries) in time $O(n^{3/2} \cdot \log^2{n})$.
\end{theorem}

\begin{algorithm}[h!]
    \caption{Approximate NE general payoffs}
    \label{algorithm:approximateNE}
    \DontPrintSemicolon
    \KwData{$\epsilon$; query access to utility function $u$ of $n$-player anonymous game $G$;
    parameters $\tau$ (failure probability), $\delta$ (accuracy of queries).}
    \KwResult{$O\left(\epsilon \right)$-NE of $G$.}
    \Begin{
        Set $\zeta=\epsilon$.
        Let $\bar{G}$ be the $\zeta$-smoothed version of $G$,
        as in \refDefinition{def:smooth_game}.\;
        \tcp*{By \refLemma{lemma:tv_utility} and \refLemma{lemma:tvd_lipschitz} it follows that\\
        $\bar{G}$ is $\lambda$-Lipschitz for $\lambda=O(1/\zeta\sqrt{n})$.}
        Apply \refAlgorithm{algorithm:approxNELip} to $\bar{G}$, simulating each all-players $\delta$-accurate
        query to $\bar{G}$ using multiple queries according to \refLemma{lemma:empirical_utility}.\;
        Let $\bar{p}$ be the obtained pure profile solution to $\bar{G}$.\;
        Construct $p$ by replacing probabilities of 0 in $\bar{p}$ with $\zeta$ and probabilities
        of 1 with $1-\zeta$.\;
        \KwRet $p$.
    }
\end{algorithm}

\begin{proof}
Set $\zeta$ equal to $\epsilon$ and let $\bar{G}$ be the
$\zeta$-smoothed version of $G$. We claim that
$\bar{G}$ is a $\lambda$-Lipschitz game for $\lambda = O\left((\zeta \sqrt{n})^{-1}\right)$.
Let $X^{(x)}_{-i}$ be as in \refDefinition{def:smooth_game}.
By \refLemma{lemma:tvd_lipschitz}, $\left\|X^{(x - 1)}_{-i} - X^{(x)}_{-i}\right\|_\TV
\leq O\left(\frac{1}{\zeta \sqrt{n}}\right)$ for all $x \in [n-1]$.
Then by \refLemma{lemma:tv_utility}, we have
\[
\left|\bar{u}^i_j(x - 1) - \bar{u}^i_j(x)\right| \leq O\left(\frac{1}{\zeta \sqrt{n}}\right).
\]

\refTheorem{thm:lipschitz2strategy} shows that
\refAlgorithm{algorithm:approxNELip} finds a pure-strategy
$3(\lambda + \delta)$-WSNE of $\bar{G}$, using $O(\log{n})$
$\delta$-accurate all-players queries.
Thus, \refAlgorithm{algorithm:approxNELip} finds a $O(\frac{1}{\zeta\sqrt{n}}+\delta)$-WSNE
of $\bar{G}$, where $\delta$ is the additive accuracy of queries.

Despite not being allowed to query $\bar{G}$ directly, we can
simulate any $\delta$-accurate query to $\bar{G}$ with a set of
randomized all-players queries to $G$. This is done in the body
of \refAlgorithm{algorithm:approximateNE}.
By \refLemma{lemma:empirical_utility}, for $\tau>0$,
$(1/2\delta^2)\log(4n/\tau)$ randomized queries to $G$ correctly simulate a
$\delta$-accurate query to $\bar{G}$ with probability $\geq 1-\tau$.

In total, the algorithm makes $O\left(\log{n} \cdot (1/\delta^2) \cdot \log(n/\tau)\right)$
all-players payoff queries to $G$. With a union bound over the $4\log{n}$
simulated queries to $\bar{G}$, this works with probability $1-4\tau\log{n}$.

Once we find this pure-strategy $O\left(\frac{1}{\zeta\sqrt{n}}+\delta\right)$-WSNE of $\bar{G}$,
the last part of \refAlgorithm{algorithm:approximateNE} maps the pure output
profile to a mixed one where whoever plays $1$ in $\bar{G}$ places
probability $(1 - \zeta)$ on $1$, and whoever plays $2$ in $\bar{G}$
places probability $\zeta$ on $1$. It is easy to verify
that the regret experienced by player $i$ (that is, the difference in payoff between
$i$'s payoff and $i$'s best-response payoff) in $G$ is at most $\zeta$ more
than the one she experiences in $\bar{G}$.

The extra additive $\zeta$ to the regret of players means that we have an
$\epsilon$-NE of $G$ with $\epsilon=O(\zeta+\delta+\frac{1}{\zeta\sqrt{n}})$.
The query complexity thus is\newline $O(\log{n} \cdot (1/\delta^2) \cdot \log(n/\tau))$.

Setting $\delta=1/\sqrt[4]{n}$, $\zeta=1/\sqrt[4]{n}$, $\tau = 1/ 16\log{n}$,
we find an $O(1/\sqrt[4]{n})$-Nash equilibrium using
$O(\sqrt{n} \cdot \log^2{n})$ all-players queries with
probability at least $3/4$. We remark that the above parameters
can be chosen to satisfy any given approximation guarantee $\epsilon \geq n^{-1/4}$,
i.e., simply find solutions to the equation $\epsilon = \zeta+\delta+(\zeta\sqrt{n})^{-1}$.
This allows for a family of algorithms parameterized by $\epsilon$, for $\epsilon \in [n^{-1/4}, 1)$,
thus an approximation scheme.

The runtime is equal to the number of single-payoff queries and
can be calculated as follows. Calculating the value of $\phi(i)$
in \refAlgorithm{algorithm:approxNELip} takes $O(n \sqrt{n} \log{n})$.
We make $O(\sqrt{n} \log{n})$ queries to $G$ to simulate one in $\bar{G}$,
and once we gather all the information, we need an additional linear time factor
to count the number of players whose best response is $1$. The fact that
the above part is performed at every step of the binary search implies a total
running time of $O(n^{3/2} \cdot \log^2{n})$ for \refAlgorithm{algorithm:approxNELip}.
\refAlgorithm{algorithm:approximateNE} simply invokes \refAlgorithm{algorithm:approxNELip}
and only needs linear time to construct the profile $p$; thus, it runs in the same time.
\end{proof}

\subsection{Lower Bound}

Finally, we consider the question of lower bounds that are complementary
to Theorem~\ref{thm:approximateNE}.
We use the minimax principle and thus define a distribution
over instances that will lead to the lower bound on query complexity,
for any deterministic algorithm.
We specify a distribution over certain games that possess
a unique pure Nash equilibrium. The $n$ players that participate in any
of these games are partitioned into $\log{n}$ groups, which are
numbered from $1$ to $\log{n}$. Group $i$'s equilibrium strategy
depends on what all the previous groups $\{1, \dots, i - 1\}$ play at
equilibrium. Hence, finding out what the last group should play leads
to a lower bound of $\Omega(\log{n})$ all-players queries.

\begin{definition}
Suppose we need to discover an unknown bit-string $A := A_1 \dots
A_d$ of length $d$. A \emph{longest-common-prefix query} (lcp) takes a bit-string $B :=
B_1 \dots B_d$ as input and outputs the length of the longest common
prefix between $A$ and $B$, i.e., lcp$(A, B):= \max_{j \in \{0, \dots, d\}} \{A_i
= B_i \text{ for all } i = 1, \dots, j\}$.
\end{definition}

\begin{lemma}
\label{lemma:bitstring_lowerbound}
Let $A := A_1 \dots A_d$ be an unknown bit-string generated uniformly
at random. Then, the expected number of queries needed by any
lcp-query algorithm is $\Omega(d)$.
\end{lemma}

\begin{proof}
We show by induction on $q$
that the expected length of the longest common prefix between $A$ and
the $q$-th queried bit-string is at most $2q$. We use $B^{(q)}$ to denote
the queried input at step $q$.

Suppose $\mathcal{A}$ makes one query. Then, the expected prefix length is
\begin{equation*}
\begin{split}
\expected[\text{lcp}(A, B^{(1)})] & = \Pr[A_1 = B_1] \cdot 1 + \dots
       + \Pr[A_1 = B_1 \wedge \dots \wedge A_d = B_d] \cdot d\\
& = \sum_{i = 1}^d \frac{i}{2^i} \leq \sum_{i = 1}^\infty \frac{i}{2^i} = 2.
\end{split}
\end{equation*}
Thus, the base case holds. Now, assume $\expected[\text{lcp}(A, B^{(q-1)})] \leq 2(q-1)$.
Moreover, let $A_1 \dots A_{2(q - 1)} = B^{(q)}_1 \dots B^{(q)}_{2(q - 1)}$. Then,
\begin{equation*}
\begin{split}
\expected[\text{lcp}(A, B^{(q)})] & = 2(q-1) + \sum_{i = 1}^{d - 2(q-1)} \frac{i}{2^i}\\
& \leq 2(q-1) + \sum_{i = 1}^\infty \frac{i}{2^i} = 2q.
\end{split}
\end{equation*}
Hence, $\mathcal{A}$ needs to make at least $d/2$ queries in
order to output a prefix of expected length equal to $d$.
\end{proof}

\begin{lemma}
\label{lemma:unique_WSCE}
Let $\mathcal{G}_n$ be the class of $n$-player two-strategy anonymous
games such that $u^i_1(x) = 1 - u^i_2(x)$ and $u^i_1(
x) \in \{0, 1\}$, for all $i \in [n], x \in \{0, \dots, n - 1\}$.
Then, there exists a distribution $\mathcal{D}_n$ over
$\mathcal{G}_n$ such that every $G$ drawn from $\mathcal{D}_n$ has a
unique (pure-strategy) $\epsilon$-WSNE.
\end{lemma}

\begin{proof}
Let $n = 2^k$, and let the first $n - 1$ players be partitioned into
sets $N_1, \dots, N_k$ such that $|N_j| = n / 2^j$ for all $j \in
[k]$. Let $I_j := \left\{0, \dots, 2^{j - 1} - 1\right\}$. Moreover, let $1 \succ_i^x
2$ denote player $i$ preferring strategy $1$ to strategy $2$ given
that $x$ other players are playing strategy $1$. We use $1 \succ_P^I
2$ to mean that $1 \succ_i^x 2$ for all $i \in P$ and all $x \in
I$. Furthermore, all players $i \in N_j$ share the same preferences.

We define $\mathcal{D}_n$ in the following manner. For all $j \in
[k]$, let $N_j$ flip $|I_j|$ fair coins, one associated to each subset
$S_{j, \ell} := \left\{\frac{\ell}{2^{j - 1}} n, \dots, \frac{\ell +
  1}{2^{j - 1}} n - 1\right\}$, with $\ell \in I_j$, to decide whether they
prefer strategy $1$ to $2$ within $S_{j, \ell}$. This means that
$N_j$'s preferences over a subset $S_{j, \ell}$ do not depend on what
$N_j$ prefer at some different $S_{j, m} \subset \{0, \dots, n -
1\}$. Player $n$, who belongs to no set $N_j$, flips a coin for every
$\ell \in \{0, \dots, n - 1\}$. In particular, with probability $1 / 2$, $N_1$
always prefer strategy $1$ to $2$, and with probability $1 / 2$, $2$
to $1$. 

Every game $G$ drawn from $\mathcal{D}_n$ has the feature of
having a unique PNE. Moreover, $N_j$'s unique best response at
equilibrium depends on what $N_1, \dots, N_{j - 1}$ play.
An application of iterated elimination of dominated strategies
suffices to verify the two claims. Further, it can be shown, e.g., by induction on $j$,
that in any $\epsilon$-WSNE ($\epsilon < 1$), no member of $N_j$ would place
positive probability onto her worst-response.
\end{proof}

\begin{theorem}\label{thm:log_lowerbound}
Let $\mathcal{G}_n$ be defined as in
\refLemma{lemma:unique_WSCE}. Then, for any $\epsilon \in [0, 1)$, any
randomized all-players query algorithm must make $\Omega(\log{n})$
queries
to find an $\epsilon$-WSNE of $\mathcal{G}_n$ in the worst case.
\end{theorem}

\begin{proof}
We use an adversarial distribution $\mathcal{D}_n$ over
$\mathcal{G}_n$ and apply the minimax principle to
lower-bound the expected cost that any deterministic algorithm must
incur. We let $\mathcal{D}_n$ be defined as in
\refLemma{lemma:unique_WSCE}, so that it has a unique $\epsilon$-WSNE,
and show that learning this equilibrium is equivalent to learning an
unknown $\log(n)$-long bit-string as in
\refLemma{lemma:bitstring_lowerbound}.

If we associate a random indicator variable $Y_j$ to $N_j$ that
is equal to 1 if and only if $N_j$ are playing strategy $1$ at
equilibrium, then $y := \sum_{j = 1}^k |N_j| \cdot Y_j$
corresponds to the number of players playing strategy $1$ that player
$n$ sees at equilibrium. It is easily verified that, by
$\mathcal{D}_n$'s definition, $Y := Y_1 \dots Y_k$ is generated
uniformly at random since every group $N_j$ flips a series of
fair coins to determine their preferences. Clearly, any algorithm
$\mathcal{A}$ that outputs an $\epsilon$-WSNE of
$\mathcal{G}_n$ is able to tell what $Y$'s value is because it simply
requires to look at what $\{1, \dots, n - 1\}$ play in the pure
equilibrium profile. Suppose an all-players query returns, for all $i
\in [n]$, both payoffs for playing strategy $1$ and $2$ against $x \in
\{0, \dots, n - 1\}$ other players playing strategy 1. This assumption
can only strengthen the query model.

Let $B := B_1 \dots B_k$ be the binary representation of a query input
$x \in \{0, \dots, n - 1\}$.
Due to every player in $N_j$ sharing the same preferences, we can
assume the answer is of the form $A := A_1 \dots A_k$, where
$A_j$ is equal to 1 if and only if $N_j$'s best response to $x$ is
strategy $1$. We now argue that an all-players query is not able to give more
information than the longest common prefix between $A$ and
$B$.

Suppose $A_1 \dots A_{\ell - 1} A_{\ell + 1} \dots A_{k} = B_1
\dots B_{\ell - 1} B_{\ell + 1} \dots B_{k}$, and $A_\ell \neq B_\ell$
for some $\ell \in \{1, \dots, k - 1\}$. According to
$\mathcal{D}_n$'s definition, $N_j$ flip a fair coin to determine
their preferences for every consecutive subset $S \subseteq \{0, \dots, n -
1\}$ of size $\frac{n}{2^{j - 1}}$. Assume, w.l.o.g., that $A_\ell
= 1$ and $B_\ell = 0$. Let $m \in \{\ell + 1, \dots, k - 1\}$.
$N_{m}$'s payoff at $x$ and $N_{m}$'s payoff at any $x'$ such that
$\left|x - x'\right| \geq \frac{n}{2^\ell}$ correspond to two independent coin
flips, i.e., it is useless to know that $B_j = A_j$ for all $j \in \{\ell + 1,
\dots, k\}$ since $N_\ell$ must play $1$.
Furthermore, any attempt to guess $Y_{\ell + 1} \dots Y_{k}$
fails with probability $1 - 2^{\ell + 1 - k}$.

For any random $\log(n)$-long bit-string $Y$
there exists a game $G$ in the support of $\mathcal{D}_n$
whose unique pure-strategy $\epsilon$-WSNE
is equal to $Y$. By
\refLemma{lemma:bitstring_lowerbound}, the expected
number of queries needed by any algorithm $\mathcal{A}$
to discover $Y$ is $\Omega(\log{n})$.
\end{proof}

\section{Conclusions and Further Work}

Our interest in the query complexity of anonymous games
has resulted in an algorithm that has an improved runtime-efficiency
guarantee, although limited to when the number of strategies $k$
is equal to 2.
\refAlgorithm{algorithm:approximateNE} (\refTheorem{thm:approximateNE}) finds
an $\epsilon$-NE faster than the PTAS of \cite{dp14},
for any $\epsilon \geq 1 / \sqrt[4]{n}$. In particular, for
$\epsilon = 1 / \sqrt[4]{n}$, their algorithm runs in
subexponential time, while ours is just $\tilde{O}(n^{3/2})$;
however, our $\epsilon$-NE is not well-supported.

An immediate question is whether we can obtain sharper
bounds on the query complexity of two-strategy games.
There are ways to potentially strengthen the results.
First, our lower bound holds for well-supported equilibria; it would be
interesting to know whether a logarithmic number of queries is also
needed to find an $\epsilon$-NE for $\epsilon < \frac{1}{2}$. 
We believe this is the case at least for small values of $\epsilon$.
Second, the $\epsilon$-NE found by our algorithm are not well-supported
since all players are forced to randomize.
Is there a query-efficient algorithm that finds an $\epsilon$-WSNE?
Third, we may think of generalizing the algorithm to the (constant)
$k$-strategy case by requiring every player to place
probability either $\frac{\zeta}{k}$ or $1 - \frac{k - 1}{k}
\zeta$ and solve an associated Lipschitz game. However, in
this case we can no longer use binary search to find a fixed point of
the smooth game. As a consequence, the query complexity
is likely to be substantially higher.

\bibliographystyle{plain}
\bibliography{literature}

\end{document}